\documentclass[a4paper,UKenglish,cleveref, autoref]{lipics-v2019}

\title{Linear transformations between colorings in chordal graphs} 

\titlerunning{Linear transformations between colorings in chordal graphs} 

\author{Nicolas Bousquet}{
Univ. Grenoble Alpes, CNRS, Grenoble INP, G-SCOP, France}{nicolas.bousquet@grenoble-inp.fr}{https://orcid.org/0000-0003-0170-0503}{}

\author{Valentin Bartier}{Univ. Grenoble Alpes, CNRS, Grenoble INP, G-SCOP, France}{valentin.bartier@grenoble-inp.fr}{}{}

\authorrunning{N. Bousquet and V. Bartier}

\Copyright{Nicolas Bousquet, Valentin Bartier}



\category{}

\supplement{}

\funding{This work was supported by ANR project GrR (ANR-18-CE40-0032)}

\nolinenumbers 

\EventEditors{John Q. Open and Joan R. Access}
\EventNoEds{2}
\EventLongTitle{42nd Conference on Very Important Topics (CVIT 2016)}
\EventShortTitle{CVIT 2016}
\EventAcronym{CVIT}
\EventYear{2016}
\EventDate{December 24--27, 2016}
\EventLocation{Little Whinging, United Kingdom}
\EventLogo{}
\SeriesVolume{42}
\ArticleNo{23}
\usepackage{stmaryrd}
\usepackage{enumitem}

\newtheorem{conjecture}[theorem]{Conjecture}
\newtheorem{question}[theorem]{Question}
\newtheorem{observation}[theorem]{Observation}

\newcommand{\G}{\mathcal{G}}

\newcommand{\bnu}{\boldsymbol{\nu}}
\newcommand{\bmu}{\boldsymbol{\mu}}

\EventEditors{Michael A. Bender, Ola Svensson, and Grzegorz Herman}
\EventNoEds{3}
\EventLongTitle{27th Annual European Symposium on Algorithms (ESA 2019)}
\EventShortTitle{ESA 2019}
\EventAcronym{ESA}
\EventYear{2019}
\EventDate{September 9--11, 2019}
\EventLocation{Munich/Garching, Germany}
\EventLogo{}
\SeriesVolume{144}
\ArticleNo{11}

\begin{document}

\maketitle

\begin{abstract}
    Let $k$ and $d$ be such that $k \ge d+2$. Consider two $k$-colorings of a $d$-degenerate graph $G$. Can we transform one into the other by recoloring one vertex at each step while maintaining a proper coloring at any step? Cereceda et al. answered that question in the affirmative, and exhibited a recolouring sequence of exponential length. 
    
    If $k=d+2$, we know that there exists graphs for which a quadratic number of recolorings is needed. And when $k=2d+2$, there always exists a linear transformation. In this paper, we prove that, as long as $k \ge d+4$, there exists a transformation of length at most $f(\Delta) \cdot n$ between any pair of $k$-colorings of chordal graphs (where $\Delta$ denotes the maximum degree of the graph). The proof is constructive and provides a linear time algorithm that, given two $k$-colorings $c_1,c_2$ computes a linear transformation between $c_1$ and $c_2$.
\end{abstract}

\section{Introduction}

Reconfiguration problems consist in finding step-by-step transformations between two feasible solutions of a problem such that all intermediate states are also feasible. Such problems model dynamic situations where a given solution already in place has to be modified for a more desirable one while maintaining some properties throughout the transformation. Reconfiguration problems have been studied in various fields such as discrete geometry~\cite{BoseLPV18}, optimization~\cite{BBRM18} or statistical physics~\cite{mohar4} in order to transform, generate, or count solutions.
In the last few years, graph reconfiguration received a considerable attention, e.g. reconfiguration of independent sets~\cite{BonamyB17,LokshtanovM18}, matchings~\cite{BousquetHIM18,ItoKK0O17}, dominating sets~\cite{SuzukiMN14} or fixed-degree sequence graphs~\cite{BousquetM18}. 
For a complete overview of the reconfiguration field, the reader is referred to the two recent surveys on the topic~\cite{Nishimura17,Heuvel13}.

Two main questions are at the core of combinatorial reconfiguration. (i) Is it possible to transform any solution into any other? (ii) If yes, how many steps are needed to perform this transformation? These two questions and their algorithmic counterparts received considerable attention.

\paragraph*{Graph recoloring.}
Throughout the paper, $G=(V,E)$ denotes a graph, $n=|V|$, $\Delta$ denotes the maximum degree of $G$, and $k$ is an integer. For standard definitions and notations on graphs, we refer the reader to~\cite{Diestel}.
A \emph{(proper) $k$-coloring} of $G$ is a function $f : V(G) \rightarrow \{ 1,\ldots,k \}$ such that, for every edge $xy\in E$, we have $f(x)\neq f(y)$. Since we will only consider proper colorings, we will then omit the proper for brevity. The \emph{chromatic number} $\chi(G)$ of a graph $G$ is the smallest $k$ such that $G$ admits a $k$-coloring.
Two $k$-colorings are \emph{adjacent} if they differ on exactly one vertex. The \emph{$k$-reconfiguration graph of $G$}, denoted by $\G(G,k)$ and defined for any $k\geq \chi(G)$, is the graph whose vertices are $k$-colorings of $G$, with the adjacency condition defined above. 
Cereceda, van den Heuvel and Johnson provided an algorithm to decide whether, given two $3$-colorings, one can transform one into the other in polynomial time and characterized graphs for which $\G(G,3)$ is connected~\cite{Cereceda09,CerecedaHJ11}.
Given any two $k$-colorings of a graph, it is $\mathbf{PSPACE}$-complete to decide whether one can be transformed into the other for $k \geq 4$~\cite{BonsmaC07}. 

The \emph{$k$-recoloring diameter} of a graph $G$ is the diameter of $\G(G,k)$ if $\G(G,k)$ is connected and is $+\infty$ otherwise. In other words, it is the minimum $D$ for which any $k$-coloring can be transformed into any other one through a sequence of at most $D$ recolorings.
Bonsma and Cereceda~\cite{BonsmaC07} proved that there exists a class $\mathcal{C}$ of graphs and an integer $k$ such that, for every graph $G \in \mathcal{C}$, there exist two $k$-colorings whose distance in the $k$-reconfiguration graph is finite and super-polynomial in $n$.


A graph $G$ is \emph{$d$-degenerate} if any subgraph of $G$ admits a vertex of degree at most $d$. In other words, there exists an ordering $v_1,\ldots,v_n$ of the vertices such that for every $i$, $v_i$ has at most $d$ neighbors in $v_{i+1},\ldots,v_n$.
It was shown independently by Dyer et al~\cite{dyer2006randomly} and by Cereceda et al.~\cite{Cereceda09} that for any $d$-degenerate graph $G$ and every $k \geq d+2$, $\G(G,k)$ is connected. Note that the bound on $k$ is the best possible since the $\G(K_n,n)$ is not connected. Cereceda~\cite{Cereceda} conjectured the following:

\begin{conjecture}[Cereceda~\cite{Cereceda}]
For every $d$, every $k \ge d+2$, and every $d$-degenerate graph $G$, the diameter of $\G(G,k)$ is at most $C_d \cdot n^2$.
\end{conjecture}

If true, the quadratic function is the best possible, even for paths, as shown in~\cite{BonamyJ12}.
Bousquet and Heinrich~\cite{BousquetH19} recently proved that the diameter of $\G(G,k)$ is $O(n^{d+1})$. In the general case,
Cereceda's conjecture is only known to be true for $d=1$ (trees)~\cite{BonamyJ12} and $d=2$ and $\Delta \le 3$~\cite{FeghaliJP16}. The diameter of $\G(G,k)$ is $O(n^2)$ when $k \ge \frac 32 (d+1)$ as shown in~\cite{BousquetH19}. Even if Cereceda's conjecture is widely open for general graphs, it has been proved for a few graph classes, e.g. chordal graphs~\cite{BonamyJ12}, bounded treewidth graphs ~\cite{BonamyB18}, and bipartite planar graphs~\cite{BousquetH19}.

Jerrum conjectured that if $k \ge \Delta+2$, the mixing time (time needed to approach the stationary distribution) of the Markov chain of graph colorings\footnote{A random walk on $\mathcal{G}(G,k)$. Fore more details on Markov chains on graph colorings, the reader is for instance referred to~\cite{ChenDMPP19}.} is $O(n \log n)$. So far, the conjecture has only been proved if $k \ge (\frac{11}{6}-\epsilon) \Delta$~\cite{ChenDMPP19}.
 Since the diameter of the reconfiguration graph is a lower bound of the mixing time, a lower bound on the diameter is of interest to study the mixing time of the Markov chain. In order to obtain such a mixing time, we need an (almost) linear diameter.
 
 The diameter of $\G(G,k)$ is linear if $k \ge 2d+2$~\cite{BousquetP16} or if $k$ is at least the grundy number of $G$ plus $1$~\cite{BonamyB18}. When $k =d+2$, the diameter of $\G(G,k)$ may be quadratic, even for paths~\cite{BonamyJ12}. But it might be true that the diameter of $\G(G,k)$ is linear whenever $k \ge d+3$. In this paper, we investigate the following question, raised for instance in~\cite{BousquetH19}: when does the $k$-recoloring diameter of $d$-degenerate graphs become linear?

\paragraph*{Our results.}
A graph is \emph{chordal} if it does not contain any induced cycle of length at least $4$. Chordal graphs admit a \emph{perfect elimination ordering}, i.e. there exists an ordering $v_1,\ldots,v_n$ of $V$ such that, for every $i$, $N[v_i] \cap \{ v_{i+1},\ldots,v_n\}$ is a clique. Chordal graphs are $(\omega(G)-1)$-degenerate where $\omega(G)$ is the size of a maximum clique of $G$. Our main result is the following:

\begin{theorem}\label{thm:chordal}
Let $\Delta$ be a fixed integer. Let $G$ be a $d$-degenerate chordal graph of maximum degree $\Delta$. For every $k \ge d+4$, the diameter of $\G(G,k)$ is at most $O_{\Delta}(n)$. Moreover, given two colorings $c_1,c_2$ of $G$, a transformation of length at most  $O_{\Delta}(n)$ can be found in linear time.
\end{theorem}

Theorem~\ref{thm:chordal} improves the best existing upper bound on the diameter of $\G(G,k)$ (where $G$ is chordal) was quadratic up to $k=2d+1$~\cite{BousquetP16}.

Note that the bound on $k$ is almost the best possible since we know that this result cannot hold for $k \le d+2$~\cite{BonamyJ12}. So there is only one remaining case which is the case $k=d+3$.


\begin{question}
Is the diameter of $\G(G,d+3)$ at most $f(\Delta(G)) \cdot n$ for any $d$-degenerate graph $G$?
\end{question}

In some very restricted cases (such as power of paths), our proof technique can be extended to $k=d+3$, but this is mainly due to the very strong structure of these graphs. For chordal graphs (or even interval graphs), we need at least $d+4$ colors at several steps of the proof and decreasing $k$ to $d+3$ seems to be a challenging problem.

We also ask the following question: is it possible to remove the dependency on $\Delta$ to only obtain a dependency on the degeneracy? More formally:

\begin{question}
Is the diameter of $\G(G,d+3)$ at most $f(d) \cdot n$ for any $d$-degenerate chordal graph $G$?
\end{question}

The best known result related to that question is the following: $\G(G,k)$ has linear diameter if $k \ge 2d+2$ (and $f$ is a constant function)~\cite{BousquetP16}. 

\begin{question}
Is the diameter of $\G(G,d+3)$ at most $f(\Delta(G)) \cdot n$ for any bounded treewidth graph $G$?
\end{question}

Our proof cannot be immediately adapted for bounded treewidth graphs since we use the fact that all the vertices in a bag have distinct colors. Feghali~\cite{article} proposed a method to "chordalize" bounded treewidth graphs for recoloring problems. However his proof technique does not work here since it may increase the maximum degree of the graph. We nevertheless think that our proof technique can be adapted in order to study many well-structured graph classes.

\paragraph*{Proof outline.}
In order to prove Theorem~\ref{thm:chordal}, we introduce a new proof technique to obtain linear diameters for recoloring graphs.
The existing results (e.g.~\cite{BousquetP16}) ensuring that $\mathcal{G}(G,k)$ has linear diameter are based on inductive proofs that completely fail when $k$ is close to $d$.
On the other hand, in the proofs giving quadratic diameters (e.g.~\cite{BonamyJ12}), the technique usually consists in finding two vertices that can be "identified" and then applying induction on the reduced graph. In that case, even if we can identify two vertices after a constant number of single vertex recolorings, we only obtain a quadratic diameter (since each vertex might "represent" a linear number of initial vertices). Both approaches are difficult to adapt to obtain linear transformations since they do not use or "forget" the original structure of the graph.

Let us roughly describe the idea of our method on interval graphs.
A buffer $\mathcal{B}$ is a set of vertices contained in $f(\omega)$ consecutive cliques of the clique path. We assume that at the left of the buffer, the coloring of the graph already matches the target coloring. We moreover assume that the coloring of $\mathcal{B}$ is special in the sense that, for every vertex $v$ in $\mathcal{B}$, at most $d+2$ colors appear in the neighborhood of $v$ ~\footnote{Our condition is actually even more restrictive.}. Note that in order to satisfy this property (and several others detailed in Section~\ref{sec:buff-def}), the buffer has to be ``long enough''. The main technical part of the proof consists in showing that if the buffer is ``long enough'', then we can modify the colors of vertices of the buffer in such a way that the same assumptions hold for the buffer starting one clique to the right of the starting clique of $\mathcal{B}$. We simply have to repeat at most $n$ times this operation to guarantee that the coloring of the whole graph is the target coloring. Since a vertex is recolored only if it is in the buffer and a vertex is in the buffer a constant (assuming $\Delta$ constant) number of times, every vertex is recolored a constant number of times.

The structure of this special coloring of the buffer, which is the main novelty of this paper, is described in Sections~\ref{ssec:buffer} to~\ref{ssec:vbuffer}. 
We actually show that this graph recoloring problem can be rephrased into a "vectorial recoloring problem" (Section~\ref{ssec:vectorial-recol}) which is easier to manipulate. And we finally prove that this vectorial recoloring problem can be solved by recoloring every element (and then every vertex of the graph) at most a constant number of times in Section~\ref{sec:interval-proof}.

\section{Buffer and vectorial coloring}\label{sec:buff-def}
Throughout this section, $G = (V,E)$ is a chordal graph on $n$ vertices of maximum clique number $\omega$ and maximum degree $\Delta$. Let $k \geq \omega + 3$ be the number of colors denoted by $1,\ldots,k$. Given two integers $x \leq y$, $\llbracket x, y \rrbracket$ is the set $\{x, x+1, \ldots, y\}$. The \emph{closed neighbourhood} of a set $S \subseteq V$ is $N[S] := S \cup (\cup_{v \in S} N(v))$.

\subsection{Chordal graphs and clique trees}\label{subsec:canonical-color}

\paragraph*{Vertex ordering and canonical coloring.}
Let $v_1, v_2, \ldots v_n$ be a perfect elimination ordering of $V(G)$. A greedy coloring of $v_n, v_{n-1}, \ldots, v_1$ gives an optimal coloring $c_0$ of $G$ using only $\omega$ colors. The coloring $c_0$ is called the \textit{canonical coloring of $G$}. The colors $c \in 1, 2 \ldots, \omega$ are called the \emph{canonical colors} and the colors $c > \omega$ are called the \emph{non-canonical colors}. Note that the independent sets $X_i := \{v \in V \text{ such that } c_0(v) = i\}$ for $i \le \omega$, called the \emph{classes} of $G$, partition the vertex set $V$.

\paragraph*{Clique tree.}
Let $G=(V,E)$ be a chordal graph. A \emph{clique tree} of $G$ is a pair $(T,B)$ where $T=(W,E')$ is a tree and $B$ is a function that associates to each node $U$ of $T$ a subset of vertices $B_U$ of $V$ (called \emph{bag}) such that: (i) every bag induces a clique, (ii) for every $x \in V$, the subset of nodes whose bags contain $x$ forms a non-empty subtree in $T$, and (iii) for every edge $(U,W) \in T$, $B_U \setminus B_W$ and $B_W \setminus B_U$ are non empty. Note that the size of every bag is at most $\omega(G)$.
A clique-tree of $G$ can be found in linear time~\cite{10.1007/3-540-60618-1_88}. Throughout this section, $T = (V_T, E_T)$ is a clique tree of $G$. We assume that $T$ is rooted on an arbitrary node. Given a rooted tree $T$ and a node $C$ of $T$, the \emph{height} of $C$ denoted by $h(C)$ is the length of the path from the root to $C$. 
\begin{observation}\label{obs:max-deg-ct}
Let $G$ be a chordal graph of maximum degree $\Delta$ and $T$ be a clique tree of $G$ rooted in an arbitrary node. Let $x$ be a vertex of $G$ and $C_i, C_j$ be two bags of $T$ that contain $x$. Then $h(C_j) - h(C_i) \leq \Delta$.
\end{observation}
\begin{proof}
We can assume without loss of generality (free to replace the one with the smallest height by the first common ancestor of $C_i$ and $C_j$) that $C_i$ is an ancestor of $C_j$ (indeed, this operation can only increase the difference of height). Let $P$ be the path of $T$ between $C_i$ and $C_j$ and $(U,W)$ be an edge of $P$ with $h(U) < h(W)$. By assumption on the clique tree, there is a vertex $y$ that appears in $B_W$ and that does not appear in $B_U$. Since this property is true for every edge of $T$ and since all the bags of $P$ induce cliques and contain $x$, the vertex $x$ has at least $|P|$ neighbors. 
\end{proof}


\subsection{Buffer, blocks and regions}\label{ssec:buffer}
Let $C_e$ be a clique of $T$. We denote by $T_{C_e}$ the subtree of $T$ rooted in $C_e$ and by $h_{C_e}(C)$ the height of the clique $C \in T_{C_e}$. Given a vertex $v \in T_{C_e}$, we say that $v$ \emph{starts} at height $h$ if the maximum height of a clique of $T_{C_e}$ containing $v$ is $h$ (in $T_{C_e}$).

Let $s := 3 {\omega \choose 2} + 2$ and $N = s + k - \omega + 1$ where $k$ is the number of colors. 
The \emph{buffer $\mathcal{B}$ rooted in $C_e$} is the set of vertices of $G$ that start at height at most $3 \Delta N - 1$ in $T_{C_e}$. For every $0 \leq i \leq 3N-1$, the \emph{block} $Q_{3N-i}$ of $\mathcal{B}$ is the set of vertices of $G$ that start at height $h$ with $i\Delta \leq h \leq (i+1)\Delta-1$. 
Finally, for $0 \leq i \leq N-1$, the \emph{region} $R_i$ of $\mathcal{B}$ is the set of blocks $Q_{3i+1}, Q_{3i+2}, Q_{3(i+1)}$. Unless stated otherwise, we will always denote the three blocks of $R_i$ by $A_i$, $B_i$ and $C_i$, and the regions of a buffer $\mathcal{B}$ by $R_1, \ldots, R_N$. Given a color class $X_p$ and $S \subseteq V$, we denote by $N[S,p]$ the set $N[S \cap X_p]$. By definition of a block and Observation \ref{obs:max-deg-ct} we have:

\begin{observation}\label{obs:separation}
Let $C_e$ be a clique of $T$ and $\mathcal{B}$ be the buffer rooted in $T_{C_e}$. Let $Q_{i-1},Q_i, Q_{i+1}$ be three consecutive blocks of $\mathcal{B}$. Then $N[Q_i] \subseteq Q_{i-1} \cup Q_i \cup Q_{i+1}$. In particular for each region $R_i = (A_i, B_i, C_i)$ of $\mathcal{B}$, $N[B_i] \subseteq R_i$.
\end{observation}
\begin{proof}
Let $v$ be a vertex of $Q_i$. By definition of $Q_i$, $v$ starts at height $h$ with $(3N-i)\Delta \leq h \leq (3N - i + 1)\Delta-1$. Let $u$ be a neighbour of $v$. By Observation \ref{obs:max-deg-ct}, $u$ starts at height $h'$ with $h-\Delta \leq h' \leq h+\Delta$, thus we have $(3N-i-1)\Delta \leq h' \leq (3N - i + 2)\Delta-1$. It follows that $u$ belongs either to $Q_{i-1}$, $Q_i$, or $Q_{i+1}$.
\end{proof}

We refer to this property as the \textit{separation property}. It implies that when recoloring a vertex of $Q_i$, one only has to show that the coloring induced on $Q_{i-1}, Q_i, Q_{i+1}$ remains proper.

\begin{figure}
    \centering
    \includegraphics[scale=0.85]{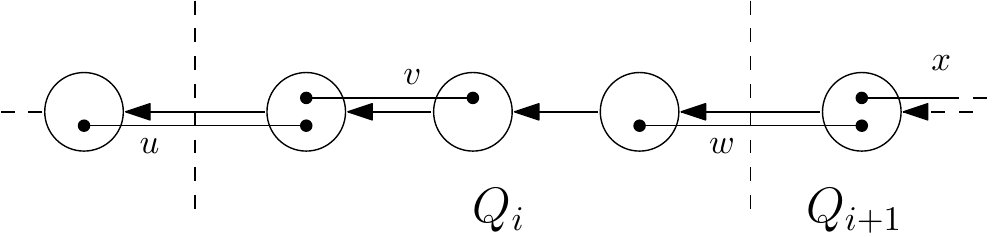}
    \label{fig:blocks}
    \caption{The nodes represent cliques of $G$. The vertices $v$ and $w$ belong to $Q_i$. The vertex $u$ does not belong to $Q_i$ even if it intersects cliques of $Q_i$, and the vertex $x$ belongs to $Q_{i+1}$.}
\end{figure}
\begin{figure}
    \centering
    \includegraphics[scale=0.85]{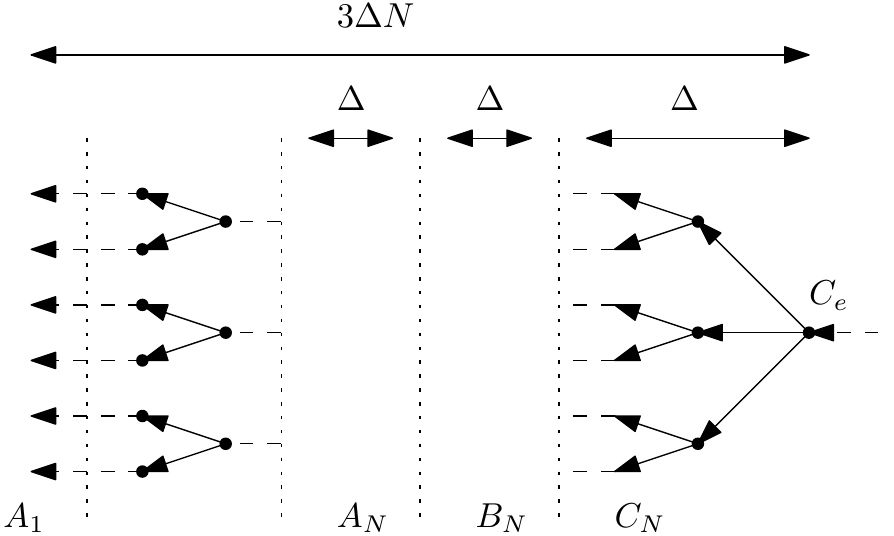}
    \label{fig:buffer}
    \caption{The buffer $\mathcal{B}$ rooted at $C_e$. The dots represent cliques of $T$ and the dashed lines separates the blocks of $\mathcal{B}$.}
\end{figure}


\subsection{Vectorial coloring}\label{ssec:vrecol}

Let $\mathcal{B}$ be a buffer. We denote the set of vertices of class $p$ that belong to the sequence of blocks $Q_i, \ldots, Q_j$ of $\mathcal{B}$ by $(Q_i, \ldots, Q_j, p)$. A \emph{color vector} $\nu$ is a vector of size $\omega$ such that $\nu(p) \in \llbracket 1, k \rrbracket$ for every $p \in \llbracket 1, \omega \rrbracket$, and $\nu(p) \neq \nu(q)$ for every $p \neq q \leq \omega$. A block $Q$ is \emph{well-colored for a color vector $\nu_Q$}  if all the vertices of $(Q,p)$ are colored with $\nu_Q(p)$. It does not imply that all the colors are $\leq \omega$ but just that all the vertices of a same class have the same color (and vertices of different classes have different colors). For brevity, we say that $(Q, \nu_{Q})$ is \emph{well-colored} and when $\nu_Q$ is clear from context we just say that $Q$ is \emph{well-colored}. In particular, a well-colored block is properly colored. Since the set $(Q, p)$ may be empty, a block may be well-colored for different vectors. However, a color vector defines a unique coloring of the vertices of a block (if $(Q, \nu_{Q})$ is well-colored then every vertex of $(Q,p)$ has to be colored with $\nu_{Q}(p)$). The color vector $\nu$ is \emph{canonical} if $\nu(p) = p$ for every $p \leq \omega$. A sequence of blocks $Q_1,\ldots,Q_r$ is \emph{well-colored for $(\nu_1,\ldots,\nu_r)$} if $(Q_i,\nu_i)$ is well-colored for every $i \le r$. 
\smallskip

\begin{definition}[Waiting regions]
A region $R$ well-colored for vectors $\nu_A, \nu_B,\nu_C$ is a \textit{waiting region} if $\nu_A = \nu_B = \nu_C$.
\end{definition}

\begin{definition}[Color region] 
A region $R$ well-colored for vectors $\nu_A, \nu_B,\nu_C$ is a \textit{color region} if there exist a canonical color $c_1$, a non-canonical color $z$ and a class $p$ such that:
\begin{enumerate}
    \item $\nu_A(m) = \nu_B(m) = \nu_C(m) \notin \{c_1,z\}$ for every $m \neq p$.
    \item $\nu_A(p) = c_1$ and $ \nu_B(p) = \nu_C(p) = z$.
\end{enumerate}
\end{definition}
In other words, the color of exactly one class is modified from a canonical color to a non-canonical color between blocks $A$ and $C$. 
We say that the color $c_1$ \emph{disappears} in $R$ and that the color $z$ \emph{appears} in $R$. For brevity we say that $R$ is a color region for the class $X_p$ and colors $c_1,z$.

\begin{definition}[Transposition region]
A region $R$ well-colored for vectors $\nu_A, \nu_B,\nu_C$ is a \textit{transposition region} if there exist two canonical colors $c_1 \neq c_2$, two non-canonical colors $z \neq z'$ and two distinct classes $p,q$ such that:
\begin{enumerate}
    \item $\nu_A(m) = \nu_B(m) = \nu_C(m) \notin \{c_1,c_2,z,z'\}$ and is canonical for every $m \notin \{p, q$\}.
    \item $\nu_A(p) = c_1$, $\nu_B(p) = z$, $\nu_C(p) = c_2$.
    \item $\nu_A(q) = c_2$, $\nu_B(q) = z'$, $\nu_C(q) = c_1$.
\end{enumerate}
\end{definition}
Note that $\nu_{A}$ and $\nu_{C}$ only differ on the coordinates $p$ and $q$ which have been permuted. The colors $z$ and $z'$ are called the \emph{temporary colors} of $R$. Note that, the coloring induced on $R$ is proper since the separation property ensures that $N[A \cap R] \subseteq A \cup B$, $N[C \cap R] \subseteq B \cap C$ and no class in $B$ is colored with $c_1$ nor $c_2$.
\smallskip

Let $\nu$ be a color vector. The color vector $\nu'$ is obtained from $\nu$ by \emph{swapping the coordinates}
$p, \ell \leq \omega$ if for every $m \notin \{p,l\}$, $\nu'(m) = \nu(m)$, $\nu'(p) = \nu(\ell)$, and $\nu'(\ell) = \nu(p)$. In other words, $\nu'$ is the vector obtained from $\nu$ by permuting the coordinates $p$ and $\ell$. \\
\emph{Swapping the coordinates $p$ and $\ell$} in a region $R$ well-colored for $(\nu_A, \nu_B, \nu_C)$ means that for every block $Q \in \{A, B, C\}$, we replace $\nu_Q$ by the color vector $\nu'_Q$ obtained by swapping the coordinates $p$ and $\ell$ of $\nu_Q$. It does not refer to a reconfiguration operation but just to an operation on the vectors.

\begin{figure}
    \centering
    \includegraphics[scale=0.75]{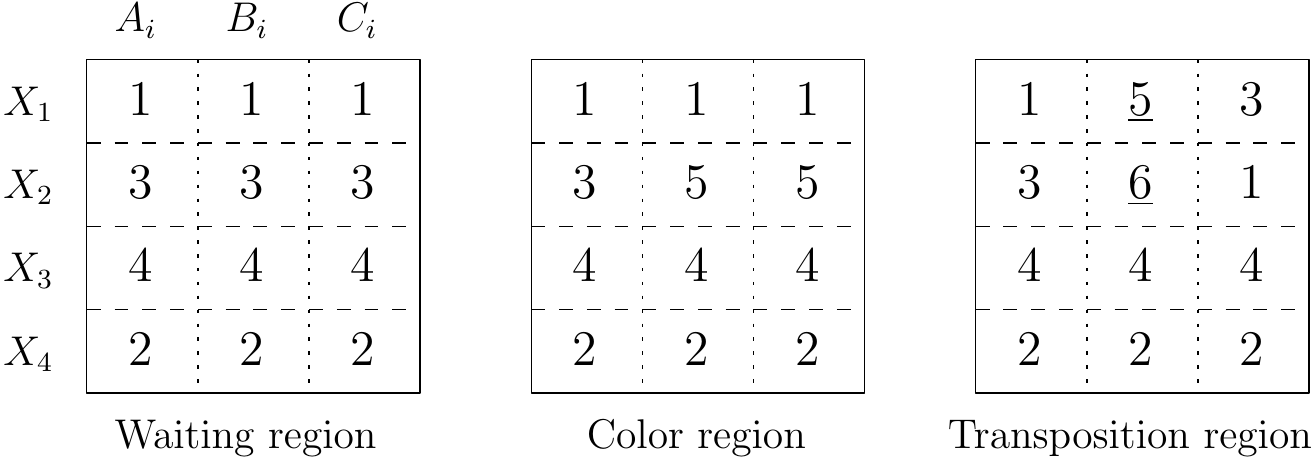}
    \caption{Example of waiting, color, and transposition regions with $\omega = 4$. Each square represents a region, the dotted lines separate the blocks and the dashed lines separate the classes. The colors $1$ to $4$ are the canonical colors and the colors $5$ and $6$ are non-canonical colors. The underlined colors in the transposition region indicate the temporary colors.}
    \label{fig:regions}
\end{figure}

\begin{observation}\label{obs:swapping-colors}
Swapping two coordinates in a waiting (resp. color, resp. transposition) region leaves a waiting (resp. color, resp transposition) region.
\end{observation}

Using the following lemma, we can assume that all the transposition regions use the same temporary colors.

\begin{lemma}\label{obs:choose-temporary}
Let $R$ be a transposition region with temporary colors $z,z'$. Let $z'' \notin \{z, z'\}$ be a non-canonical color. By recoloring the vertices of $R$ at most once, we can assume the temporary colors are $z, z''$.
\end{lemma}
\begin{proof}
Let $p$ and $q$ be the coordinates which are permuted in $R$. By definition of transposition regions, no vertex of $R$ is colored with $z''$. As any class is an independent set and by the separation property, we can recolor $(B, q)$ with $z''$ to obtain the desired coloring of $R$. 
\end{proof}

\subsection{Valid and almost valid buffers}\label{ssec:vbuffer}

In what follows, a bold symbol $\bnu$ always denote a tuple of vectors and a normal symbol $\nu$ always denotes a vector. Let $\mathcal{B} = R_1, R_2, \ldots, R_N$ be a buffer such that all the regions $R_i = A_i, B_i, C_i$ are well-colored for the vectors $\nu_{A_i}, \nu_{B_i}, \nu_{C_i}$. So $\mathcal{B}$ is well-colored for $\bnu=(\nu_{A_1},\nu_{B_1},\nu_{C_1},\nu_{A_2},\ldots,\nu_{C_N})$. The buffer $(\mathcal{B}, \bnu)$ is \textit{valid} if: 
\begin{enumerate}
    \item\label{property-valid-buffer:continuity}[\emph{Continuity property}] For every $i \in 1,2,...,N-1$, $\nu_{C_i} = \nu_{A_{i+1}}$.
    \item\label{property-valid-buffer:first-canonical} The vectors $\nu_{A_1}, \nu_{B_1}$ and $\nu_{C_1}$ are canonical (and then $R_1$ is a waiting region).
    \item\label{property-valid-buffer:transp-buf} The regions $R_2, \ldots, R_{s-1}$ define a \textit{transposition buffer}, that is a sequence of consecutive regions that are either waiting regions or transposition regions using the same temporary colors.
    \item\label{property-valid-buffer:col-buf} The regions $R_{s+1}, \ldots, R_{N-1}$ define a \textit{color buffer}, that is a sequence of consecutive regions that are either color regions or waiting regions.
    \item\label{property-valid-buffer:waiting} The regions $R_{s}$ and $R_N$ are waiting regions.
\end{enumerate}
Note that Property \ref{property-valid-buffer:continuity} along with the definition of well-colored regions enforce "continuity" in the coloring of the buffer: the coloring of the last block of $R_i$ and the first block of $R_{i+1}$ in a valid buffer have to be the same. 

An \emph{almost valid buffer} $(\mathcal{B},\bnu)$ is a buffer that satisfies Properties 1 to 4 of a valid buffer and for which Property $5$ is relaxed as follows: 
\begin{enumerate}
    \item[5'.]\label{property-almost-buffer:waiting} The region $R_{s}$ is a transposition region or a waiting region. Regions $R_1$ and $R_{N}$ are waiting regions.
\end{enumerate}

Let us make a few observations.
\begin{observation}\label{obs:t-buffer-1}
Let $(\mathcal{B}, \bnu)$ be an almost valid buffer. For every $i \leq s$, the color vectors $\nu_{A_i}$ and $\nu_{C_i}$ are permutations of the canonical colors. 
\end{observation}
\begin{proof}
By induction on $i$. By Property \ref{property-valid-buffer:first-canonical} of almost valid buffers, it is true for every $i \leq r$. Suppose now that the property is verified for $R_i$ with $2 < i < s$. By assumption $\nu_{C_i}$ is a permutation of $\llbracket 1,\omega \rrbracket$ and the continuity property (Property \ref{property-valid-buffer:continuity} of almost valid buffer) ensures that $\nu_{A_{i+1}} = \nu_{C_i}$. By Property \ref{property-valid-buffer:transp-buf} we only have two cases to consider, either $R_{i+1}$ is a waiting region and by definition $\nu_{c_{i+1}} = \nu_{A_{i+1}}$, or $R_{i+1}$ is a transposition region. In the later case, by definition of a transposition region, $\nu_{C_{i+1}}$ is equal to $\nu_{A_{i+1}}$ up to a transposition of some classes $k, \ell$ and thus is a permutation of the canonical colors.
\end{proof}


\begin{observation}\label{obs:c-buffer-1}
Let $(\mathcal{B}, \bnu)$ be an almost valid buffer and $c$ be a non-canonical color. There exists a unique class $p \leq \omega$ such that $\nu_{C_s}(p) = c$. Furthermore, either the class $p$ is colored with $c$ on all the blocks of $R_{s+1}, \ldots, R_N$, or the color $c$ disappears in a color region for the class $p$.
\end{observation}
\begin{proof}
By Observation \ref{obs:t-buffer-1}, $\nu_{C_s}$ is a permutation of the canonical colors. Thus there exists a unique class $p \leq \omega$ such that $\nu_{C_s}(p) = c$. Furthermore, by Property \ref{property-valid-buffer:col-buf} of almost valid buffer, the regions $R_{s+1}, \ldots, R_N$ are either waiting or color regions. The continuity property then ensures that either the class $p$ is colored with $c$ on $R_{s+1},\ldots, R_N$ or that the color $c$ disappears in a color region if there exists a color region for the class $p$.
\end{proof}

Since only non-canonical colors can appear in a color region, we have the following observation:

\begin{observation}\label{obs:c-buffer-1-bis}
Let $(\mathcal{B}, \bnu)$ be an almost valid buffer and $z$ be a non-canonical color. Either no vertex of $R_{s+1}, \ldots, R_N$ is colored with $z$, or there exists a color region $R_i$ with $s < i < N$ for the class $p$ in which $z$ appears. In the latter case, the vertices of the color buffer of $\mathcal{B}$ colored with $z$ are exactly the vertices of $(B_i, C_i, \ldots, C_N, p)$.
\end{observation}

Finally, since the number of regions in the color buffer is the number of non-canonical colors, we have:

\begin{sloppypar}
\begin{observation}\label{obs:non-canonical-waiting}
Let $(\mathcal{B}, \bnu)$ be an almost valid buffer. There exists a waiting region in $R_{s+1}, \ldots, R_{N-1}$ if and only if there exists a non-canonical color that does not appear in $R_{s+1}, \ldots, R_{N}$.
\end{observation}
\end{sloppypar}

\subsection{Vectorial recoloring}\label{ssec:vectorial-recol}

Let $(\mathcal{B}, \bnu)$ be a buffer. The tuple of color vectors $\bnu = (\nu_{Q_1}, \ldots, \nu_{Q_{3 \Delta N}})$ is a \emph{(proper) vectorial coloring} of $\mathcal{B}$ if for every color $c$ and every $i \leq 3 \Delta N - 1$ such that $c$ is in both $\nu_{Q_{i}}$ and $\nu_{Q_{i+1}}$, then there exists a unique class $p \leq \omega$ such that $\nu_{Q_i}(p) = \nu_{Q_{i+1}}(p) = c$.

\begin{observation}
Any proper vectorial coloring $(\mathcal{B}, \bnu)$ induces a proper coloring of $G[\mathcal{B}]$.
\end{observation}
\begin{proof}
Indeed, two different classes in two consecutive blocks cannot have the same color in a proper vectorial coloring. Since for any block $Q_i$ of $\mathcal{B}$ and for any class $p$, $N[Q_i, p] \subseteq Q_{i-1} \cup Q_i \cup Q_{i+1}$, the coloring induced on $G[\mathcal{B}]$ is proper.
\end{proof}

Note that if $(\mathcal{B}, \bnu)$ is an almost valid buffer then $\bnu$ is a proper vectorial coloring of $\mathcal{B}$ by the continuity property and the definition of waiting, color, and transposition regions. Since we will only consider proper vectorial colorings, we will omit the term proper for brevity.

Let $\nu_{Q}$ be a color vector. A color vector $\nu'_{Q}$ is \emph{adjacent to} $\nu_{Q}$ if there exists a class $p$ and a color $c \notin \nu_{Q}$ such that $\nu'_{Q}(p) = c$ and $\nu'_{Q}(m) = \nu_{Q}(m)$ for every $m \neq p$. 

\begin{observation}\label{obs:vector-recol-block}
Let $Q$ be a block well-colored for $\nu_Q$ and let $\nu'_{Q}$ be a color vector adjacent to $\nu_Q$ such that $\nu'_{Q}(p) = c \neq \nu_{Q}(p)$. Then recoloring the vertices of $(Q,p)$ one by one is a proper sequence of recolorings of $G[Q]$ after which $Q$ is well-colored for $\nu'_{Q}$.
\end{observation}

Let $(\bnu, \bnu')$ be two vectorial colorings of a buffer $\mathcal{B}$. The coloring $\bnu'$ is a \emph{vectorial recoloring} of $\bnu$ if there exists a unique $i \in \llbracket 1, 3 \Delta N \rrbracket$ such that $\nu'_{Q_i}$ is adjacent to $\nu_{Q_i}$ and $\nu'_{Q_j} = \nu_{Q_j}$ for $j \neq i$.
By Observation \ref{obs:vector-recol-block}, we have:

\begin{observation}\label{obs:vector-recol-buff}
Let $t \geq 1$ and $(\bnu^1, \bnu^t)$ be two (proper) vectorial colorings of a buffer $\mathcal{B}$. If there exists a sequence of adjacent (proper) vectorial recolorings $\bnu^1, \bnu^2, \ldots \bnu^t$, then there exists a sequence of (proper) single vertex recolorings of $G[\mathcal{B}]$ after which the coloring of $\mathcal{B}$ is well-colored for $\bnu^t$.
\end{observation}

Given a sequence of vectorial recolorings $\bnu^1, \bnu^2, \ldots \bnu^t$, we say that each coordinate is recolored at most $\ell$ times if for every coordinate $p \leq \omega$ and every $r \in \llbracket 1, 3\Delta N\rrbracket$, there exist at most $\ell$ indices $t_1,\ldots,t_\ell$ such that the unique difference between $\bnu^{t_i}$ and $\bnu^{t_{i+1}}$ is the p-th coordinate of the r-th vector of the tuples.

%

\section{Algorithm outline}\label{sec:interval-proof}

Let $G$ be a chordal graph of maximum degree $\Delta$ and maximum clique size $\omega$, $T$ be a clique tree of $G$, and $\phi$ be any $k$-coloring of $G$. We propose an iterative algorithm that recolors the vertices of the bags of $T$ from the leaves to the root until we obtain the canonical coloring defined in Section \ref{subsec:canonical-color}. Let $S$ be a clique of $T$. A coloring $\alpha$ of $G$ is \emph{treated up to $S$} if:
\begin{enumerate}
    \item Vertices starting at height more than $3 \Delta N$ in $T_S$ are colored canonically, and
    \item The buffer rooted at $S$ is valid.
\end{enumerate}
Let $C$ be a clique of $T$.
We associate a vector $\nu_C$ of length $\omega$ to the clique $C$ as follows. We set $\nu_{C}(\ell) = \alpha(v)$ if there exists $v \in X_{\ell} \cap C$. Then we arbitrarily complete $\nu_C$ in such a way all the coordinates of $\nu_C$ are distinct (which is possible since $|\nu_C| < k$). 

Given two vectors $\nu$ and $\nu'$ the \emph{difference} $D(\nu, \nu')$ between $\nu$ and $\nu'$ is $|\{p: \nu(p) \neq \nu'(p)\}|$, i.e. the number of coordinates on which $\nu$ and $\nu'$ differ. Given an almost valid buffer $(\mathcal{B}, \bnu)$ and a vector $\nu_C$, the \emph{border error} $D_{\mathcal{B}}(\nu_C, \bnu)$ is $D(\nu_{C_N}, \nu_C)$. 

Let $\mathcal{B}$ be a buffer. The class $p \leq \omega$ is \emph{internal} to $\mathcal{B}$ if $N[R_N, p] \subseteq R_{N-1} \cup R_N$. 

We first state the main technical lemmas of the paper with their proof outlines and finally explain how we can use them to derive Theorem~\ref{thm:chordal}.

\begin{lemma}\label{lem:step-1}
Let $C$ be a clique associated with $\nu_C$. Let $S$ be a child of $C$, $\mathcal{B}$ be the buffer rooted at $S$ and $\bnu$ be a tuple of vectors such that $(\mathcal{B}, \bnu)$ is valid. If $D_{\mathcal{B}}(\nu_C, \bnu) > 0$, then there exists a recoloring sequence of $\cup_{i = s}^N R_i$ such that the resulting coloring $\bnu'$ satisfies $D_{\mathcal{B}}(\nu_C, \bnu') < D_{\mathcal{B}}(\nu_C, \bnu)$, and $(\mathcal{B}, \bnu')$ is almost valid.  \\
Moreover, every coordinate  of $\cup_{i = s}^N R_i$ is recolored at most $3$ times and only internal classes are recolored.
\end{lemma}
\begin{proof}[Outline of the proof]
Let $\ell$ be a class on which $\nu_C$ and $\nu_{C_N}$ are distinct. Then, in particular, no vertex of $X_\ell$ is in $C_N \cap C$ thus the class $\ell$ is internal. Given an internal class $\ell$, if we modify $\nu_{C_N}(\ell)$ and maintain a proper vectorial coloring of the buffer $\mathcal{B}$, then the corresponding recoloring of the graph is proper. So, if we only recolor internal classes of $R_N$, then we simply have to check that the vectorial coloring of $\mathcal{B}$ remains proper. The proof is then based on a case study depending on whether $\nu_C(\ell)$ is canonical or not. A complete proof is given in Section \ref{ssec:pf-step1}.
\end{proof}

\begin{lemma}\label{lem:step-2}
Let $(\mathcal{B}, \bnu)$ be an almost valid buffer. There exists a recoloring sequence of $\cup_{i = 2}^s R_i$ such that every coordinate is recolored at most $6$ times and the resulting coloring $\bnu'$ is such that $(\mathcal{B}, \bnu')$ is valid. 
\end{lemma}
\begin{proof}[Outline of the proof]
The proof distinguishes two cases: either there exists a waiting region in the transposition buffer or not. In the first case, we show that we can "slide" the waiting regions to the right of the transposition buffer and then ensure that $R_s$ is a waiting region. Otherwise, because of the size of the transposition buffer, then some pair of colors has to be permuted twice. In this case, we show that these two transpositions can be replaced by waiting regions (and we can apply the first case). A complete proof is given in Section \ref{ssec:pf-step2}.
\end{proof}

Note that given a clique $C$ and its associated vector $\nu_C$, applying Lemma \ref{lem:step-2} to an almost valid buffer $(\mathcal{B}, \bnu)$ rooted at a child $S$ of $C$ does not modify $D_{\mathcal{B}}(\nu_C, \bnu)$ since the region $R_N$ is not recolored.

Let $C$ be a clique and $S_1,S_2$ be two children of $C$. For every $i \le 2$, let $\mathcal{B}_i$ be the buffer of $S_i$ and assume that $\mathcal{B}_i$ is valid for $\bnu^i$. We say that $\mathcal{B}_1$ and $\mathcal{B}_2$ have \emph{the same coloring} if $\bnu^1=\bnu^2$.

\begin{lemma}\label{lem:step-3}
Let $C$ be a clique associated with $\nu_C$. Let $S_1, S_2, \ldots S_e$ be the children of $C$, and for every $i \le e$, $\mathcal{B}_i$ be the buffer rooted at $S_i$. Let $\bnu^i$ be a vectorial coloring such that $(\mathcal{B}_i, \bnu^i)$ is valid. If $D_{\mathcal{B}_i}(\nu_C, \bnu^i) = 0$ for every $i \leq e$, then there exists a recoloring sequence of $\cup_{j = 2}^{N-1} R^i_j$ such that every coordinate is recolored $O(\omega^2)$ times, the final coloring of all the $\mathcal{B}_i$s is the same coloring $\bnu'$, $D_{\mathcal{B}_i}(\nu_c, \bnu') = 0$, and $(\mathcal{B}_i, \bnu')$ is valid for every $i \leq e$.
\end{lemma}

\begin{proof}[Outline of the proof]
First, we prove that it is possible to transform the coloring of $\mathcal{B}_i$ in such a way that all the color buffers have the same coloring, and that $\nu_s^1=\nu_s^i$ for $i \in \llbracket2, e \rrbracket$. \\
We then have to ensure that the vectors of the transposition buffers are the same, which is more complicated. Indeed, even if we know that the vectors $\nu_s^1$ and $\nu_s^i$ are the same, we are not sure that we use the same sequence of transpositions in the transposition buffers of $\mathcal{B}_1$ and $\mathcal{B}_i$ to obtain it. Let $\tau_1,\ldots,\tau_r$ be the set of transpositions of $\mathcal{B}_1$. The proof consists in showing that we can add to $\mathcal{B}_i$ the transpositions $\tau_1,\ldots,\tau_r,\tau^{-1}_r,\ldots,\tau_1^{-1}$ at the beginning of the transposition buffer. 
It does not modify $\nu_{A_s}$ since this sequence of transpositions gives the identity. Finally, we prove that $\tau^{-1}_r,\ldots,\tau_1^{-1}$ can be cancelled with the already existing transpositions of $\mathcal{B}_i$. And then the transposition buffer of $\mathcal{B}_i$ only consists of $\tau_1,\ldots,\tau_r$. A complete proof is given in Section \ref{ssec:pf-step3}.
\end{proof}

\begin{lemma}\label{lem:step-4}
Let $C$ be a clique of $T$ with children $S_1, S_2, \ldots S_e$ and let $\alpha$ be a $k$-coloring of $G$ treated up to $S_i$ for every $i \in \llbracket 1 , e \rrbracket$. Let $\nu_C$ be a vector associated with $C$ and $\mathcal{B}_i = R^i_{1}, \ldots, R^i_{N}$ denote the buffer rooted at $S_i$. Assume that there exists $\bnu$ such that $(\mathcal{B}_i,\bnu)$ is valid and satisfies $D_{\mathcal{B}_i}(\nu_C, \bnu) = 0$ for every $i \leq e$. Then there exists a recoloring sequence of $\cup_{j = 2}^{N-1}R^i_j$ such that, for every $i \leq e$, every vertex of $\mathcal{B}_i$ is recolored at most one time and such that the resulting coloring of $G$ is treated up to $C$. A complete proof is given in Section \ref{ssec:pf-step4}.
\end{lemma}

\begin{proof}[Outline of the proof]
This proof "only" consists in shifting the buffer of one level. We simply recolor the vertices that now start in another region (of the buffer rooted at $C$) with their new color. We prove that the recoloring algorithm cannot create any conflict. A complete proof of the lemma is given in Section~\ref{ssec:pf-step4}.
\end{proof}

Given Lemmas \ref{lem:step-1}, \ref{lem:step-2}, \ref{lem:step-3} and \ref{lem:step-4} we can prove our main result:

\begin{theorem}\label{thm:1}
Let $\Delta$ be a fixed constant. Let $G(V,E)$ be a $d$-degenerate chordal graph of maximum degree $\Delta$ and $\phi$ be any $k$-coloring of $G$ with $k \geq d + 4$. Then we can recolor $\phi$ into the canonical coloring $c_0$ in at most $ O (d^4 \Delta \cdot n)$ steps. Moreover the recoloring algorithm runs in linear time.
\end{theorem}
\begin{proof}
Let $c_0$ be the canonical coloring of $G$ as defined in Section \ref{subsec:canonical-color}, and $T$ be a clique tree of $G$. Let us first show that given a clique $C \in T$ with children $S_1, \ldots, S_e$ and a coloring $\alpha$ treated up to $S_i$ for every $i \leq e$, we can obtain a coloring of $G$ treated up to $C$. Let $\nu_C$ be a vector associated with $C$. For every $i \le e$, let $\mathcal{B}_i$ be the buffer rooted in $S_i$ and $\bnu_i$ be a vectorial coloring of $\mathcal{B}_i$ such that $(\mathcal{B}_i,\bnu^i)$ is valid. 
For every $i \leq e$, by applying Lemmas~\ref{lem:step-1} and~\ref{lem:step-2} at most $D_{\mathcal{B}_i}(\nu_C, \bnu^i)$ times to $(\mathcal{B}_i, \bnu^i)$, we obtain a vectorial coloring $\bnu^i$ such that $(\mathcal{B}_i, \bnu^i)$ is valid and $D_{\mathcal{B}_i}(\nu_c, \bnu^i) = 0$. By Lemma \ref{lem:step-3}, we can recolor each $\bnu^i$ into $\bnu'$ such that for every $i$, $(\mathcal{B}_i, \bnu')$ is valid  and $D_{\mathcal{B}_i}(\nu_c, \bnu') = 0$. Then we can apply Lemma \ref{lem:step-4} to obtain a coloring of $G$ such that the buffer $(\mathcal{B}, \bnu)$ rooted in $C$ is valid. Since no vertex starting in cliques $W \in T_C$ with $h_{C}(W) > 3 \Delta N$ is recolored, these vertices remain canonically colored and the resulting coloring of $G$ is treated up to $C$.
Note that only vertices of $T_C$ that start in cliques of height at most $3 \Delta N$ are recolored at most $O(\omega^2)$ times to obtain a coloring treated up to $C$.
\smallskip

Let us now describe the recoloring algorithm and analyze its running time. We root $T$ at an arbitrary node $C_r$ and orient the tree from the root to the leaves. We then do a breadth-first-search starting at $C_r$ and store the height of each node in a table $h$ such that $h[i]$ contains all the nodes of $T$ of height $i$. Let $i_h$ be the  height of $T$. We apply Lemmas \ref{lem:step-1} to \ref{lem:step-4} to every $C \in h[i]$ for $i$ from $i_h$ to $0$. Let us show that after step $i$, the coloring of $G$ is treated up to $C$ for every $C \in h[i]$. It is true for $i = i_h$ since for any $C \in h[i_h]$ the sub-tree $T_C$ of $T$ only contains $C$. Suppose it is true for some $i > 0$ and let $C \in h[i-1]$. Let $S_1, \ldots, S_e$ be the children of $C$. For all $j \in 1, \ldots, e$, $S_j \in h[i]$ and by assumption the current coloring is treated up to $S_j$ after step $i$. Thus we can apply Lemmas \ref{lem:step-1} to \ref{lem:step-4} to $C$. After iteration $i_h$ we obtain a coloring of $G$ that is treated up to $C_r$. Up to adding "artificial" vertices to $G$, we can assume that $C_r$ is the only clique of $T$ adjacent to a clique path of length $3 \Delta N$ (in fact we only need a tuple of $3N$ color vectors) in $T$ and apply Lemmas \ref{lem:step-1} to \ref{lem:step-4} until we obtain a coloring such that $C_r$ is canonically colored, and the algorithm terminates. A clique tree of $G$ can be computed in linear time~\cite{10.1007/3-540-60618-1_88}, as well as building the table $h$ via a breadth-first-search. Given a clique $C$, we can access to the cliques of the buffer rooted at $T_C$ in constant time by computing their height and using the table $h$. Furthermore, a vertex of height $i$ is recolored during the iterations $i+1, \ldots, i+ 3 \Delta N$ only. As each vertex is recolored at most $O(\omega^2)$ times at each iteration, it follows that the algorithm runs in linear time. Finally, as $N = 3{{\omega}\choose{2}} + k - \omega + 3$, each vertex is recolored at most $O(\omega^4 \Delta)$ times, and thus the algorithm recolors $\phi$ to $c_0$ in at most $O(\omega^4 \Delta \cdot n)$ steps.
\end{proof}

The proof of Theorem \ref{thm:chordal} immediately follows:
\begin{proof}[Proof of Theorem~\ref{thm:chordal}]
Let $\phi$ and $\psi$ be two $k$-colorings of $G$ with $k \geq d + 4$ and let $c_0$ be the canonical coloring of $G$ defined in Section \ref{subsec:canonical-color}. By Theorem \ref{thm:1}, there exists a recoloring sequence from $\phi$ (resp. $\psi$) to $c_0$ of length $O((d+1)^4 \Delta \cdot n)$. Thus there exists a sequence of length $O(n)$ that recolors $\phi$ to $\psi$. Furthermore the recoloring sequences from $\phi$ to $c_0$ and from $\psi$ to $c_0$ can be found in linear time by Theorem \ref{thm:1}, which concludes the proof.
\end{proof}

\section{Proof of Lemma \ref{lem:step-1}}\label{ssec:pf-step1}

Recall that in a buffer $\mathcal{B}$, $R_{s}$ is the region that sits between the transposition buffer and the color buffer. Before proving Lemma \ref{lem:step-1}, let us first start with some observations.

\begin{observation}\label{obs:maintain-nature-color-waiting-2}
Let $R$ be a well-colored region that does not contain color $c$. If we set $\nu_Q(p) = c$ for every block $Q$ of $R$, then we obtain a waiting region if $R$ was a waiting region or a color region for the class $p$, and we obtain a color region if $R$ was a color region for a class $q \neq p$.
\end{observation}

Hence, modifying the color of a class on all the blocks of a region $R_j$ with $j \geq s$ of an almost valid buffer leaves a waiting region or a color region and maintains Property \ref{property-valid-buffer:col-buf} of almost valid buffers. The next lemma ensures that under technical assumptions the continuity property (Property \ref{property-valid-buffer:continuity} of valid buffers) can also be kept.

\begin{lemma}\label{obs:maintain-continuity}
Let $(\mathcal{B}, \bnu)$ be an almost valid buffer and $R_i, R_j$ be two regions of $\mathcal{B}$ with $1 \leq i \leq j$. Let $X$ be a block in $\{B_{i}, C_{i}\}$ and let $Y = B_j$ or $Y = C_N$. Then recoloring $(X, \ldots, Y, p)$ with $c$ preserves the continuity property.
\end{lemma}

\begin{proof}
Since $(\mathcal{B}, \bnu)$ is almost valid, for every $p \in \llbracket 1, \omega \rrbracket$ and every $t \in \llbracket 1 ,  N-1 \rrbracket$, we have $\nu_{C_t}(p) = \nu_{A_{t+1}}(p)$. As the sequence of recolored blocks starts in $\{B_i, C_i\}$ and ends either in $B_j$ or in $C_N$, after the recoloring we still have $\nu_{C_t}(p) = \nu_{A_{t+1}}(p)$ for every $t \in \llbracket 1 ,  N-1 \rrbracket$ and $p \in \llbracket 1 ,  \omega \rrbracket$ since both do not change or both are $c$.
\end{proof}

Given a color $c$, a class $p$, and a sequence of consecutive blocks $(Q_i, \ldots, Q_j)$, we say that $(Q_i, \ldots, Q_j, p)$ is \emph{$c$-free} if no vertex of $N[\cup_{t = i}^j Q_t \cap X_p]$ is colored with $c$. Note that a sequence of (proper) vectorial recolorings of a buffer $(\mathcal{B}, \bnu)$ that does not recolor $A_1$ and such that all the classes recolored on $C_N$ are internal to $\mathcal{B}$ yields a (proper) sequence of single vertex recolorings of $G$ by Observation~\ref{obs:vector-recol-buff}. With the definition of clique-tree we can make the following observation:

\begin{observation}\label{lem:neighbour-prop}
Let $C$ be a clique associated with $\nu_C$, $S$ be a child of $C$, and $(\mathcal{B}, \bnu)$ be the buffer rooted at $S$. If $\nu_{C_N}(\ell) \neq \nu_C(\ell)$ for some $\ell \leq \omega$, then the class $\ell$ is internal to $\mathcal{B}$.
\end{observation}
\begin{proof}
Suppose that the class $\ell \leq \omega$ is not internal to $\mathcal{B}$. Then there exists a vertex $u \in X_{\ell} \cap R_N$ which has a neighbor $v$ that does not belong to $R_{N-1} \cup R_N$. Thus $u$ and $v$ must be contained in a clique $C'$ that is an ancestor of $C$ and then by property (ii) of clique tree, $u$ must be contained in $C$. Then, by definition of $\nu_C$, it must be that $\nu_C(\ell) = \nu_{C_n}(\ell)$.
\end{proof}

We also need the following technical lemma:
\begin{lemma}\label{obs:create-cancel-color}
Let $(\mathcal{B}, \bnu)$ be an almost valid buffer. Let $s < i < N$, $c$ be a color and $p$ be an internal class. If one of the following holds:
\begin{enumerate}
    \item\label{obs:create-cancel-color-1} $R_i$ is a waiting region, $c$ is a non-canonical color that does not appear in $R_s, \ldots, R_N$ and the class $p$ is not involved in a color region, or
    \item\label{obs:create-cancel-color-2} $R_i$ is a color region for the class $p$. Moreover $c$ is non-canonical and does not appear in $R_s, \ldots, R_N$, or
    \item\label{obs:create-cancel-color-3} $R_i$ is a color region for the class $p$ where $c$ is the canonical color that disappears.
\end{enumerate}
Then changing the color of $(B_i, \ldots, C_N, p)$ by $c$ also gives an almost valid buffer and a proper coloring of $G$.
\end{lemma}
\begin{proof}
Let $\bnu'$ be the resulting coloring. As the class $p$ is internal and $c$ does not appears in $R_s, \ldots R_N$, Case \ref{obs:create-cancel-color-1} and Case \ref{obs:create-cancel-color-2} describe a proper recoloring sequence. As $(\mathcal{B}, \bnu)$ is almost valid, Observation \ref{obs:c-buffer-1} ensures that the only class colored with $c$ in $R_s, \ldots, R_N$ in Case \ref{obs:create-cancel-color-3} is the class $p$ on the set of blocks $(A_s, \ldots, A_i)$. Since the class $p$ is internal, case \ref{obs:create-cancel-color-3} also defines a proper recoloring sequence.

Let us now show that $(\mathcal{B}, \bnu')$ is almost valid. First note that no block of $R_1, \ldots, R_s$ is recolored thus Properties~\ref{property-valid-buffer:first-canonical} and~\ref{property-valid-buffer:transp-buf} of almost valid buffers are still satisfied. The recolored blocks $(B_i, \ldots, C_N)$ satisfy the condition of Lemma~\ref{obs:maintain-continuity} thus Property \ref{property-valid-buffer:continuity} also holds. Observation \ref{obs:maintain-nature-color-waiting-2} ensures that the regions $R_{i+1}, \ldots, R_N$ remain either waiting or color regions and, in particular, that $R_N$ remains a waiting region. Along with the fact that $R_s$ is not modified, Property \hyperref[property-almost-buffer:waiting]{5'} is satisfied. We finally have to check Property \ref{property-valid-buffer:col-buf}. By Observation \ref{obs:maintain-nature-color-waiting-2},  the regions $R_{i+1}, \ldots, R_{N-1}$ remain either waiting or color regions as well as regions $R_{s+1}, \ldots, R_{i-1}$ that are not modified. Let us concentrate on $R_i$. Note that in the three cases, $\nu_{Q}(m) = \nu'_{Q}(m)$ for every $m \neq p$ and $Q$ block of $R_i$, as $R_i$ is either a waiting region or a color region for the class $p$. We now have to distinguish the three cases:

\begin{itemize}
    \item In Case \ref{obs:create-cancel-color-1}, let $c_1 := \nu_{C_s}(p)$. By definition of almost valid buffers, $\nu_{A_i}(p) = c_1$ and is canonical. After the recoloring, we have $\nu'_{A_i}(p) = c_1$ and $\nu'_{B_i}(p) = \nu'_{C_i}(p) = c$. Thus $R_i$ becomes a color region in $\bnu'$ for the class $p$ where $c_1$ disappears and $c$ appears. Since by hypothesis $c$ does not appear in $R_{s+1}, \ldots R_N$ in $\bnu$ and since there is no color region for the class $p$ in $\bnu$, there exists exactly one color region for the class $p$ and colors $c_1,c$ in $\bnu'$ and Property \ref{property-valid-buffer:col-buf} follows. 
    
    \item In Case \ref{obs:create-cancel-color-2}, let $c' := \nu'_{A_i}(p)$. Since $R_i$ is a color region for the class $p$ in $\bnu$ where $c$ disappears, $c'$ is canonical. As $c$ is non-canonical $R_i$ becomes a color region for the class $p$ and colors $c',c$ in $\bnu'$. Since $\bnu$ is almost valid, $R_i$ is the only color region for the class $p$ in $\bnu$. Furthermore, the color $c$ does not appear in $\bnu$. Thus, $R_i$ is the only color region for the class $p$ and colors $c', c$ in $\bnu'$ and Property \ref{property-valid-buffer:col-buf} follows.
    
    \item In Case \ref{obs:create-cancel-color-3}, we have $\nu'_{A_i}(p) = \nu'_{B_i}(p) = \nu'_{C_i}(p) = c$. Since $R_i$ is a color region for the class $p$ in $\bnu$ it becomes a waiting region in $\bnu'$ and again Property \ref{property-valid-buffer:col-buf} follows.
\end{itemize}
\end{proof}


We can now give the proof of Lemma \ref{lem:step-1}:
\begin{proof}[Proof of Lemma~\ref{lem:step-1}]
Let $C$ be the clique associated with vector $\nu_C$ and $S$ be a child of $C$. Let $(\mathcal{B}, \bnu)$ be the valid buffer rooted at $S$. Assume that $D_{\mathcal{B}}(\nu_c, \bnu) > 0$. Then there exists $p \leq \omega$ such that $\nu_{C_N}(p) \neq \nu_C(p) := c$ and by Observation \ref{lem:neighbour-prop} the class $p$ is internal to $\mathcal{B}$. \\
The following sequences of recolorings only recolor blocks of $R_s, \ldots, R_N$ and all the recolorings fit in the framework of Lemma~\ref{obs:maintain-continuity}. Thus Properties \ref{property-valid-buffer:continuity},~\ref{property-valid-buffer:first-canonical} and~\ref{property-valid-buffer:transp-buf} of almost valid buffers are always satisfied. We then only have to check Properties~\ref{property-valid-buffer:col-buf} and~\hyperref[property-almost-buffer:waiting]{5'} to conclude the proof. Let us distinguish several cases:

\begin{figure}
\begin{subfigure}{.5\textwidth}
  \includegraphics[scale=0.7]{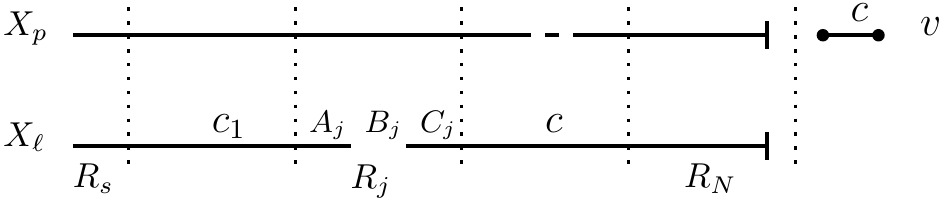}
  \caption{Case 3}
  \label{fig:step-1-case-3}
\end{subfigure}%
\begin{subfigure}{.5\textwidth}
  \includegraphics[scale=0.7]{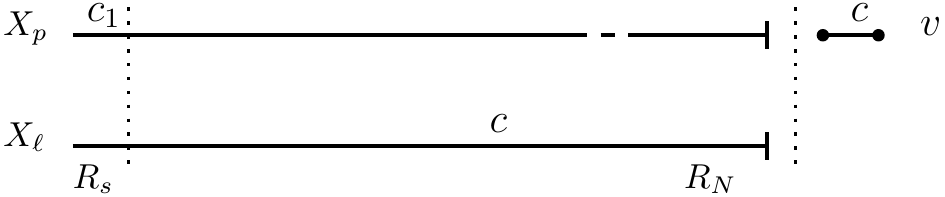}
  \caption{Case 4}
  \label{fig:step-1-case-4}
\end{subfigure}
\begin{subfigure}{.5\textwidth}
  \centering
  \includegraphics[scale=0.7]{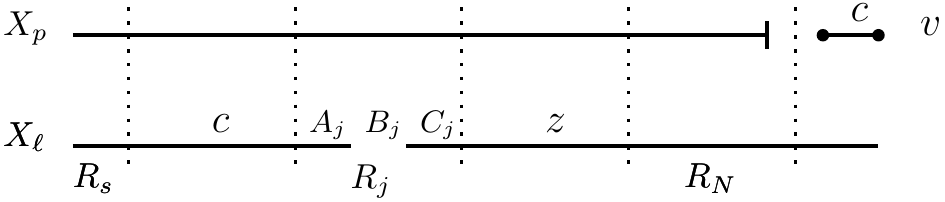}
  \caption{Case 5}
  \label{fig:step-1-case-5}
\end{subfigure}
\caption{The initial coloring $\bnu$ for cases $3$, $4$, and $5$ in the proof of Lemma~\ref{lem:step-1}. The rows represent the classes. A blank indicates a color region for the class, a dashed line indicates that the class may or may not be involved in a color region. The vertical segment at the end of the buffer indicates that the class is internal. The dotted vertical lines separate the different regions.}
\label{fig:step-1}
\end{figure}

\begin{enumerate}
    \item No class is colored with $c$ on $R_{s}, \ldots, R_N$ in $\bnu$. Then $c$ is not canonical since $\nu_{A_s}$ is a permutation of the canonical colors by Observation \ref{obs:t-buffer-1}. Suppose first that there does not exist a color region for the class $p$ in $\bnu$. Since $c$ does not appear in the color buffer, Observation \ref{obs:non-canonical-waiting} ensures that there exists a waiting region $R_i$ with $s < i < N$. Then by Lemma \ref{obs:create-cancel-color}.\ref{obs:create-cancel-color-1}, we can recolor $(B_i, \ldots, C_N, p)$ with $c$ and obtain an almost valid buffer. Suppose otherwise that there exists a color region $R_j$ for the class $p$ in $\bnu$. By Lemma \ref{obs:create-cancel-color}.\ref{obs:create-cancel-color-2}, we can recolor $(B_j, \ldots, C_N, p)$ with $c$ and obtain an almost valid buffer. In both cases, the border error decreases.
    
    \item The class $p$ is colored with $c$ on a sequence of consecutive blocks of $R_s, \ldots, R_N$. As $c \neq \nu_{C_N}(p)$, Observation~\ref{obs:c-buffer-1} ensures that $c$ and disappears in a color region $R_i$ for the class $p$ and that $c$ is canonical. By Lemma~\ref{obs:create-cancel-color}.\ref{obs:create-cancel-color-3}, we can recolor $(B_i, \ldots, C_N, p)$ with $c$ and obtain an almost valid buffer where the border error has decreased.
    
    \item A class $\ell \neq p$ is colored with $c$ and $c$ is not canonical (see Figure~\ref{fig:step-1-case-3} for an illustration).
    Since $(\mathcal{B}, \bnu)$ is valid, $c$ appears in a color region $R_j$ with $s < j < N$ for the class $\ell$ and a canonical color $c_1$. By Observation~\ref{obs:c-buffer-1-bis}, we have $\nu_{C_N}(\ell) = \nu_C(p) = c$, thus $\nu_{C_N}(\ell) \neq \nu_C(\ell)$ which implies by Observation \ref{lem:neighbour-prop} that the class $\ell$ is internal. Thus, by Lemma \ref{obs:create-cancel-color}.\ref{obs:create-cancel-color-3}, we can recolor $(B_j, \ldots, C_N, \ell)$ with $c_1$ and obtain an almost valid coloring $\bnu^{tmp}$. Since $(\mathcal{B}, \bnu)$ is valid and $R_s$ is not recolored, $(\mathcal{B}, \bnu^{tmp})$ is valid. Moreover, since $\nu_{C_N}(\ell) \neq \nu_C(\ell)$, $D_{\mathcal{B}}(\nu_C, \bnu^{tmp}) \leq D_{\mathcal{B}}(\nu_C, \bnu)$. As no class is colored with $c$ in $\bnu^{tmp}$, we can apply case 1 to $(\mathcal{B}, \bnu^{tmp})$. 
    
    \item A class $\ell \neq p$ is colored with $c$, $c$ is canonical and does not disappear in the color buffer  (see Figure~\ref{fig:step-1-case-4} for an illustration). Since $c$ is canonical, the class $\ell$ is colored with $c$ on $R_s, \ldots, R_N$ by Observation \ref{obs:c-buffer-1}. Since $\nu_{C_N}(\ell) = c = \nu_C(p) \neq \nu_C(\ell)$ the class $\ell$ is internal. \\
    Let $z, z'$ be two non-canonical colors and let $c_1 = \nu_{A_s}(p)$. Suppose first there exists a color $y$ that is not contained in $B_s, C_s, R_{s+1}, \ldots, R_N$ in $\bnu$. Then we apply the following recolorings: 
    \begin{enumerate}[label=(\arabic*)]
        \item Recolor $(B_s, p)$ with $z$ and $(B_s, \ell)$ with $z'$,
        \item Recolor $(C_s,\ldots, C_N, \ell)$ with $y$,
        \item Recolor $(C_s,\ldots, C_N, p)$ with $c$,
        \item Recolor $(C_s,\ldots, C_N, \ell)$ with $c_1$.
    \end{enumerate}
    If such a color $y$ does not exist then in particular all the non-canonical colors appear in $(\mathcal{B}, \bnu)$. Note that in that transformation we do not assume that $y$ is not canonical. Since $k \geq \omega+3$ and every non-canonical color appears in at most one class of the color buffer, there exists a class $q \notin \{p, \ell\}$ and a color region $R_j$ where some canonical color $y$ disappears and the non-canonical color $z''$ appears. Free to modify $z$ and $z'$, we can assume that $z'' \notin \{z, z'\}.$ Then we can add the following recolorings to the sequence: 
    \begin{enumerate}[label=(\arabic*)]
        \item[(0)] Recolor $(B_s, \ldots, A_j, q)$ with $z''$,
        \setcounter{enumi}{4}
        \item[(5)] Recolor $(B_s, \ldots, A_j, q)$ with $y$.
    \end{enumerate}
    And after recoloring $0$, the color $y$ is not contained anymore in $B_s,C_s, R_{s+1}, \ldots, R_N$. So we can apply recolorings $1$ to $4$. Let us justify that the colorings are proper. Recoloring $0$ is proper since $q$ is the only class that contains $z''$ in $R_s, \ldots, R_N$ by Observation \ref{obs:c-buffer-1-bis}. Recoloring $1$ is proper by the separation property and the fact that after recoloring $0$, the only non-canonical color in region $R_s$ is $z'' \notin \{z, z'\}$ (if exists). Recolorings $2$, $3$ and $4$ are proper since the classes $p$ and $\ell$ are internal and, by the separation property, after recoloring $0$, $(C_s,\ldots, C_N, \ell)$ is $y$-free, after recolorings $1$ and $2$ $(C_s,\ldots, C_N, p)$ is $c$-free and after recolorings $1$ and $3$ $(C_s,\ldots, C_N, \ell)$ is $c_1$-free. Finally recoloring $5$ is proper since after recoloring $4$, the only class colored with $y$ on $R_s,\ldots, R_N$ is the class $q$ on $A_s$.\\
    Let us show Properties~\ref{property-valid-buffer:col-buf} and~\hyperref[property-almost-buffer:waiting]{5'} of almost valid buffers are satisfied. First note the colorings $0$ and $5$ cancel each other thus the coloring is only modified on classes $p$ and $\ell$. Furthermore recolorings $3$ and $4$ consist in swapping the coordinates $p$ and $\ell$ on the regions $R_{s+1}, \ldots, R_N$. By Observation \ref{obs:swapping-colors}, the regions $R_{s+1}, \ldots, R_N$ remain either waiting or color regions (in particular $R_N$ remains a waiting region). As furthermore no color region is created Property \ref{property-valid-buffer:col-buf} holds. Finally note that in $\bnu'$, $R_s$ is a transposition region and then Property \hyperref[property-almost-buffer:waiting]{5'} follows.
    
    \item A class $\ell \neq p$ is colored with $c$, $c$ is canonical, and $c$ disappears in the color region $R_j$ for class $\ell$ where the non-canonical color $z$ appears  (see Figure~\ref{fig:step-1-case-5} for an illustration). Let $z' \neq z$ be a non-canonical color and let $c_1 = \nu_{A_s}(p)$. We apply the following recolorings:
    \begin{enumerate}[label=(\arabic*)]
        \item Recolor $(B_s, \ldots, A_j, \ell)$ with $z$,
        \item Recolor $(B_s, p)$ with $z'$,
        \item Recolor $(C_s,\ldots, C_N, p)$ with $c$,
        \item Recolor $(C_s, \ldots, A_j, \ell)$ with $c_1$.
    \end{enumerate}
    Recoloring $1$ is proper since Observation \ref{obs:c-buffer-1-bis} ensures that the only class colored with $z$ in $R_s, \ldots, R_N$ is the class $\ell$ on $(B_j, \ldots, C_N)$. By the separation property, $(B_s, \ldots, A_j, \ell)$ is $z$-free in $\bnu$. Recoloring $2$ is proper as the only non-canonical color in $R_s$ after recoloring $1$ is $z \neq z'$. Recoloring $3$ is proper as the class $p$ is internal and thus after recoloring $1$, $(C_s,\ldots, C_N, p)$ is $c$-free by the separation property. Finally, recoloring $4$ is proper as after recolorings $2$ and $3$, $(C_s, \ldots, A_j, \ell)$ is $c_1$-free by the separation property.
    
    Let us show that the resulting coloring defines an almost valid buffer. First note that regions $R_{j+1}, \ldots, R_N$ are only modified by recoloring $3$ and Observation \ref{obs:maintain-nature-color-waiting-2} ensures they remain either waiting or color regions. In particular $R_N$ remains a waiting region. Note that after recolorings $3$ and $4$, $\bnu'$ on regions $R_{s+1}, \ldots, R_{j-1}$ is obtained from $\bnu$ by swapping coordinates $p$ and $\ell$ (on these regions). Thus the nature of these regions is maintained by Observation \ref{obs:swapping-colors}. Since $R_j$ was a color region for the class $\ell$ and colors $c,z$ in $\bnu$, $B_j,C_j$ are not modified and since $\nu'_{A_j}(\ell)=c_1$, $R_j$ is a color region for class $\ell$ and colors $c_1, z$ in $\bnu'$. So the regions $R_{s+1}, \ldots, R_N$ remain either waiting or color regions. Furthermore no new color region is created and colors $c_1, z$ are involved in exactly one color region thus Property \ref{property-valid-buffer:col-buf} is satisfied.
    Finally, $R_s$ is indeed a transposition region in $\bnu'$ since $\bnu$ is a valid buffer and $z$ and $z'$ are non-canonical colors, thus Property \hyperref[property-almost-buffer:waiting]{5'} holds.
    \end{enumerate}
    
    Note that by Lemma \ref{obs:choose-temporary} we can always suppose that up to a recoloring the temporary colors $z, z'$ used in $R_s$ are the same than the ones used in the transposition buffer of $(\mathcal{B}, \bnu')$. So in every case we are able to decrease the border error of $(\mathcal{B}, \bnu)$ by one by recoloring each coordinate of $R_s, \ldots, R_N$ at most three times and obtain an almost valid buffer $(\mathcal{B}, \bnu')$, which completes the proof.
\end{proof}

\section{Proof of Lemma \ref{lem:step-2}}\label{ssec:pf-step2}
The proof distinguishes two cases:
\smallskip

\noindent\textbf{Case 1:} there is a region of the transposition buffer of $\mathcal{B}$ that is a waiting region.\\ 
The core of the proof is the following lemma:

\begin{lemma}\label{lem:transp-shift}
Let $(\mathcal{B}, \bnu)$ be an almost valid buffer and $R_i, R_{i+1}$ be two consecutive regions with $1 < i < s$ such that $R_i$ is a waiting region and $R_{i+1}$ is a transposition region. Then there exists a recoloring sequence of $R_i \cup R_{i+1}$ such that, in the resulting coloring $\bnu'$, $R_i$ is a transposition region, $R_{i+1}$ is a waiting region, and $(\mathcal{B}, \bnu')$ is almost valid. Moreover only coordinates of $R_i \cup R_{i+1}$ are recolored at most twice.
\end{lemma}

\begin{proof}
Let $p,\ell$ be the classes permuted in $R_{i+1}$. Let $c_1 = \nu_{A_{i+1}}(\ell) = \nu_{C_{i+1}}(p)$ and $c_2 = \nu_{A_{i+1}}(p) = \nu_{C_{i+1}}(\ell)$ and $z,z'$ be the temporary colors of $R_{i+1}$. Recall that since $i < s$, Observation \ref{obs:t-buffer-1} ensures that $R_i$ only contains canonical colors. We apply the following recolorings:
\begin{enumerate}
    \item Recolor $(B_i, C_i, A_{i+1}, p)$ with $z$,
    \item Recolor $(B_i, C_i, A_{i+1}, \ell)$ with $z'$,
    \item Recolor $(C_i, A_{i+1}, B_{i+1}, p)$ with $c_1$,
    \item Recolor $(C_i, A_{i+1}, B_{i+1}, \ell)$ with $c_2$.
\end{enumerate}
Since $R_i$ is a waiting region that only contains canonical colors and $R_{i+1}$ is a transposition region for classes $p$ and $\ell$ in $\bnu$, recolorings $1$ and $2$ are proper by the separation property. After the first two recolorings, the color $c_1$ (resp. $c_2$) is only contained in $(A_i, \ell)$ and $(C_{i+1}, p)$ in $R_i \cup R_{i+1}$ (resp. $(A_i, p)$ and $(C_{i+1}, \ell)$). Thus the separation property ensures that recolorings $3$ and $4$ are proper.

One can easily check that in $\bnu'$, $R_i$ is a transposition region for classes $p$ and $\ell$ and that $R_{i+1}$ is a waiting region. Let us finally prove that $(\mathcal{B},\nu')$ is almost valid.
Since $\nu'_{A_i}=\nu_{A_i}$, $\nu'_{C_{i+1}}=\nu_{C_{i+1}}$ and $\nu'_{C_{i}}=\nu_{A_{i+1}}$, the continuity property holds. The other properties are straightforward  since we only recolor the regions $R_i, R_{i+1}$ which are waiting or transposition regions in $\bnu'$.
\end{proof}

By assumption there exists a waiting region in $R_2, \ldots R_{s-1}$. Amongst all these regions let $R_i$ with $1 < i < s$ be the one with the largest index. We iteratively apply Lemma \ref{lem:transp-shift} to the regions $(R_i,R_{i+1}), (R_{i+1}, R_{i+2}), \ldots,(R_{s-1}, R_{s})$. Each class is recolored at most four times and the resulting coloring $\bnu'$ defines a valid buffer. 
\smallskip

\noindent\textbf{Case 2:} All the regions of the transposition buffer of $\mathcal{B}$ are transposition regions. \\
As there are $3{{\omega}\choose{2}}$ regions in the transposition buffer and only ${{\omega}\choose{2}}$ distinct transpositions of $\llbracket 1,\omega \rrbracket$, there must exist two distinct regions $R_i$ and $R_j$ with $1 < i < j < s$ for which the same pair of colors is transposed (note that the colors might be associated to different classes in $R_i$ and $R_j$ but it does not matter). Let us prove the following lemma:
\begin{lemma}\label{lem:transp-cancel}
Let $(\mathcal{B}, \bnu)$ be an almost valid buffer and $c_1, c_2$ be two canonical colors. If there exist two transposition regions $R_i$ and $R_j$  where colors $c_1$ and $c_2$  are transposed, then there exists a sequence of recolorings of $\cup_{t = i}^j R_t$ such that each coordinate is recolored at most twice, $R_i$ and $R_j$ are waiting regions in the resulting coloring $\bnu'$, and $(\mathcal{B}, \bnu')$ is almost valid.
\end{lemma}
\begin{proof}
Let $p,\ell$ (resp. $p'$, $\ell'$) be the classes permuted in $R_i$ (resp. $R_{j}$). Without loss of generality, we can assume that $i < j$, $\nu_{A_i}(p) = \nu_{C_i}(\ell) = \nu_{A_j}(\ell') = \nu_{C_j}(p') = c_1$ and $\nu_{A_i}(\ell) = \nu_{C_i}(p) = \nu_{A_j}(p') = \nu_{C_j}(\ell') = c_2$. By Property \ref{property-valid-buffer:transp-buf} of almost valid buffers, all the transposition regions use the same temporary colors. Let $z,z'$ be these colors and let $z'' \notin \{z, z'\}$ be another non-canonical color, which exists since $k \geq \omega + 3$. Note that $z''$ does not appear in $R_1,\ldots, R_s$. Let $I_1$ (resp. $I_2$) be the set of blocks of $C_i, \ldots, A_j$ that contains color $c_1$ (resp. $c_2$). For each block $Q \in I_1$ (resp. $Q' \in I_2$), there exists a class $p_Q$ (resp. $\ell_{Q'}$) such that $(Q, p_Q)$ (resp. $(Q',\ell_{Q'})$) is colored with $c_1$ (resp. $c_2$). We apply the following recolorings:
\begin{enumerate}
    \item For every $Q \in I_1$ recolor $(Q, p_Q)$ with $z''$,
    \item For every $Q \in I_2$ recolor $(Q, \ell_Q)$ with $c_1$,
    \item For every $Q \in I_1$ recolor $(Q_, p_Q)$ with $c_2$,
    \item Recolor $(B_i, p)$ and $(B_j, p')$ with $c_1$,
    \item Recolor $(B_i, \ell)$ and $(B_j, \ell')$ with $c_2$.
\end{enumerate}
Let us justify that the recolorings are proper. Recoloring 1 is proper since $(R_1,\ldots,R_s)$ is $z''$-free and all the blocks of $Q$ were colored with $c_1$ in $\bnu$. 
After recoloring $1$, $(C_i, \ldots, A_j)$ is $c_1$-free since $B_i$ and  $B_j$ does not contain $c_1$ since $R_i$ and $R_j$ are transposition regions in $\bnu$ for the color $c_1$. Then recoloring $2$ is proper. Recoloring $3$ is proper since, after recoloring $2$, $(C_i, \ldots, A_j)$ is $c_2$-free. At this point, the only class that contains $c_1$ in $R_i$ (resp. $R_j$) is the class $p$ (resp. $p'$) and the only class that contains $c_2$ in $R_i$ (resp. $R_j$) is the class $\ell$ (resp. $\ell'$). Thus, the recolorings $4$ and $5$ are proper. Note that each coordinate of $R_i, \ldots, R_j$ is recolored at most twice in this sequence.

Let us show that $(\mathcal{B}, \bnu')$ is an almost valid buffer. Properties \ref{property-valid-buffer:first-canonical}, \ref{property-valid-buffer:col-buf} and \hyperref[property-almost-buffer:waiting]{5'} are indeed satisfied since $R_1$ and $R_s, \ldots, R_N$ have not been recolored. By Observation \ref{obs:t-buffer-1}, for every $t \in \llbracket i, j-1 \rrbracket$, the blocks $C_t, A_{t+1}$ contain all the canonical colors, and thus belong to $I_1 \cap I_2$. As the color $c_1$ has been replaced by the color $c_2$ and conversely, and as no other color is modified on these blocks, the continuity property holds.
Let us prove Property \ref{property-valid-buffer:transp-buf} of almost valid buffers holds.
For every $t$ such that $i < t < j$, the coloring of $R_t$ in $\bnu'$ is obtained from the coloring of $R_t$ in $\bnu$ by replacing the color $c_1$ by the color $c_2$ and conversely. One can easily check that since $c_1$ and $c_2$ are canonical, $R_t$ remains either a waiting region or a transposition region after this operation.
Finally we have that in $R_i$ (resp. $R_j$), the classes $p$ and $\ell$ (resp. $p'$ and $\ell'$) are colored with the same colors $c_1$ and $c_2$ on the three blocks of the region. Thus the regions $R_i$ and $R_j$ are waiting regions and Property \ref{property-valid-buffer:transp-buf} is satisfied.
\end{proof}
The coloring $\bnu'$ of $\mathcal{B}$ obtained after applying Lemma \ref{lem:transp-cancel} is almost valid and $R_j$ is a waiting region for $\bnu'$. So case $1$ applies, which concludes the proof of Lemma \ref{lem:step-2}.

\section{Proof of Lemma \ref{lem:step-3}}\label{ssec:pf-step3}

Let $C$ be a clique associated with a vector $\nu_C$ and $S_1, \ldots, S_e$ be the children of $C$. We denote the buffer of $S_i$ by $\mathcal{B}_i = R^i_1, \ldots, R^i_N$ and its vectorial coloring by $\bnu^i)$. By assumption, for every $i \in \llbracket 1, e \rrbracket $, $(\mathcal{B}_i, \bnu^i)$ is valid and $D_{\mathcal{B}}(\nu_C, \bnu^i) = 0$. \\
The proof is divided in two main steps. We will first show that there exist sequences of recolorings such that the color buffer of all the $\mathcal{B}_i$s have the same coloring and match with the coloring of $C$. We will then show that, after this first step, there exist sequences of recolorings of the transposition buffers of the $\mathcal{B}_i$s such that the buffers all have the same coloring. \smallskip

\noindent\textbf{Step 1.} Agreement on the color buffers \\
Let $R_i$ be a color region of an almost valid buffer where the color $c$ disappears and $z$ appears. The following lemma shows that, up to recoloring each coordinate of the color buffer at most once, we can "choose" the index of the region in which $c$ disappears and $z$ appears.

\begin{lemma}\label{lem:move-col-r}
Let $(\mathcal{B}, \bnu)$ be a valid buffer and $s < i,j < N$ be such that $R_i$ is a color region for a class $p \leq \omega$ and the colors $c, z$. There exists a sequence of recolorings of $R_{s+1}, \ldots, R_{N-1}$ such that each coordinate is recolored at most once and in the resulting coloring $\bnu'$, the region $R_j$ is a color region for the class $p$ and colors $c,z$ and $(\mathcal{B}, \bnu')$ is valid. Furthermore, if $R_j$ was a waiting region in $\bnu$ then $R_i$ is a waiting region in $\bnu'$ and if $R_j$ was a color region for the colors $c',z'$ in $\bnu$ then $R_i$ is a color region for the colors $c', z'$ in $\bnu'$. Other regions remain either waiting region or color regions for the same pair of colors. 

\end{lemma}

\begin{proof}
We have two cases to consider, either $R_j$ is a waiting region or is a color region for a class $\ell \neq p$. 
\smallskip

\noindent\textit{Case 1.} $R_j$ is a waiting region. 
In that case we simply recolor $(B_i, \ldots, A_j, p)$ with $c$ if $i < j$ or we recolor $(B_j, \ldots, A_i, p)$ with $z$ if $i > j$.
Since $(\mathcal{B}, \bnu)$ is valid, Observation \ref{obs:c-buffer-1} ensures that the class $p$ is the only class colored with $c$ in $R_s, \ldots, R_N$. So the recoloring is proper. By Observation \ref{obs:maintain-nature-color-waiting-2}, the regions $R_{i+1}, \ldots, R_{j-1}$ are still waiting or color regions for the same pairs of colors in $\bnu'$. One can then easily check that the region $R_i$ is a waiting region in $\bnu'$, that $R_j$ is a color region for the class $p$ and colors $c,z$ in $\bnu'$ and that $(\mathcal{B}, \bnu')$ is valid.
\smallskip
\noindent\textit{Case 2.} $R_j$ is a color region for the class $\ell \neq p$ and colors $c', z'$. \\
Since the recoloring will be symmetric, we assume that $i < j$. We apply the following recolorings:
\begin{enumerate}
    \item Recolor $(B_i, \ldots, A_j, \ell)$ with $z'$
    \item Recolor $(B_i, \ldots, A_j, p)$ with $c$
\end{enumerate}
Since $(\mathcal{B},\bnu)$ is valid, Observations \ref{obs:c-buffer-1} and \ref{obs:c-buffer-1-bis} ensures the class $p$ (resp. $\ell$) is the only class colored with $c$ and $z$ (resp. $c'$ and $z'$) in $R_{s}, \ldots, R_N$. Thus the two recolorings are proper. The regions of $R_{s+1}, \ldots, R_{i-1}$ and $R_{j+1}, \ldots, R_{N-1}$ are not recolored, and  by Observation \ref{obs:maintain-nature-color-waiting-2}, the regions $R_{i+1}, \ldots, R_{j-1}$ are either waiting or color regions for the same pairs of colors in $\bnu'$. One can then easily check that regions $R_i$ is a color region for the class $\ell$ and colors $c', z'$, that region $R_j$ is a color region for the class $p$ and colors $c, z$ in $\bnu'$ and that $(\mathcal{B}, \bnu')$ is valid. 
\end{proof}

The following lemma will permit to guarantee that the colorings agree on the vector $\nu_{C_s}$.

\begin{lemma}\label{lem:step-3-same-col-buf}
Let $C$ be a clique associated with a vector $\nu_C$. Let $S_1, S_2$ be two children of $C$ and $\mathcal{B}_i$ be the buffer rooted at $S_i$ for $i \le 2$. Assume moreover that $(\mathcal{B}_1, \bnu)$ and $(\mathcal{B}_2, \bmu)$ are valid and satisfy $D_{\mathcal{B}_1}(\nu_C, \bnu) = D_{\mathcal{B}_2}(\nu_C, \bmu)= 0$. Then there exists a sequence of recolorings of $\cup_{j = s}^{N - 1} R^2_j$ such that in the resulting coloring $\bmu'$ of $\mathcal{B}_2$ we have $\nu_{C_s}(p) = \mu'_{C_s}(p)$ for every $p \leq \omega$, and that $(\mathcal{B}_2, \bmu')$ is valid.  Moreover each coordinate is recolored at most $12(k-\omega)$ times.   
\end{lemma}

\begin{proof}
Note that, since $D(\nu_C, \bnu) = D(\nu_C, \bmu)= 0$, if $\nu_{C_N}(p)$ is canonical then there is not color regions for the class $p$ in $\bnu$ or $\bmu$ and we have $\nu_{C_s}(p)=\nu_{C_N}(p)=\mu_{C_N}(p)= \mu_{C_s}(p)$. So if they differ on $C_s$ for some class $p$, there exists a color region for $p$. Assume that there exists a class $p$ such that $\mu_{C_s}(p) \neq \nu_{C_s}(p) := c$. By Observation~\ref{obs:t-buffer-1}, $\mu_{C_s}$ an $\nu_{C_s}$ are permutations of $\{1,\ldots,\omega \}$.
So there exists $\ell \neq p$ such that $\mu_{C_s}(\ell) = c$.  Since $\nu_{C_N}(p) = \mu_{C_N}(p)$ and $\nu_{C_N}(\ell) = \mu_{C_N}(\ell)$, there exists color regions for the classes $p$ and $\ell$ in $(\mathcal{B}_1, \bnu)$ and $(\mathcal{B}_2, \bmu)$. \\
By Lemma \ref{lem:move-col-r} we can assume that the color region for the class $p$ and color $c', z$ is $R_{s+1}$ and that the color region for the class $\ell$ and colors $c, z'$ is $R_{s+2}$.
Since $k \geq \omega + 3$, there exists a non-canonical color $z'' \notin \{z, z'\}$ such that $z''$ does not appear in $R^2_s, \ldots R^2_{s+2}$ in $\bmu$. We apply the following recolorings:
\begin{enumerate}
    \item Recolor $(B^2_s, \ldots, A^2_{s+2}, \ell)$ with $z''$,
    \item Recolor $(B^2_s, p)$ with $z'$,
    \item Recolor $(C^2_s, A^2_{s+1}, p)$ with $c$,
    \item Recolor $(C^2_s, \ldots, A^2_{s+2}, \ell)$ with $c'$.
\end{enumerate}
Let us call $\bmu''$ the resulting coloring.
Recoloring $1$ is proper since $(B_s^2, \ldots, A_{s+2}^2)$ is $z''$-free in $\bmu$. Since $(\mathcal{B}_2, \bmu)$ is valid, the only non-canonical color in $R_s$ after recoloring $1$ is $z'' \neq z$, thus recoloring $2$ is proper. After recoloring $1$, $(C^2_s, \ldots, A^2_{s+1})$ is $c$-free and after recoloring $3$, $(C^2_s, \ldots, A^2_{s+2})$ is $c'$-free thus recolorings $3$ and $4$ are proper. 

Let us show that $(\mathcal{B}_2, \bmu'')$ is almost valid. Properties \ref{property-valid-buffer:first-canonical} and \ref{property-valid-buffer:transp-buf} are indeed satisfied and Lemma~\ref{obs:maintain-continuity} ensures the continuity property holds. Let us check Property \ref{property-valid-buffer:col-buf}. The region $R^2_{s+1}$ (resp. $R^2_{s+2}$) is a color region for the class $p$ (resp. $\ell$) and colors $c, z'$ (resp. $c', z$) in $\bmu$ and thus becomes a color region for colors $(c, z')$ (resp. $c', z$) after recoloring $3$ (resp. $4$) in $\bmu''$ and Property \ref{property-valid-buffer:col-buf} follows. Finally one can easily check that the region $R_s^2$ becomes a transposition region for the classes $p$ and $\ell$. As the region $R_N^2$ is not recolored, Property \hyperref[property-almost-buffer:waiting]{5'} follows and $(\mathcal{B}_2, \bmu'')$ is almost valid. By Lemma \ref{obs:choose-temporary}, we can assume that the temporary colors used in $R_s$ are the same that the ones used in the transposition buffer of $\mathcal{B}_2$, free to recolor the coordinates of $\mu''_{B_s}$ at most twice. We can thus apply Lemma \ref{lem:step-2} to $(\mathcal{B}_2, \bmu'')$ to obtain a coloring $\bmu'$ such that $(\mathcal{B}_2, \bmu')$ is a valid buffer.

Let us count how many times a coordinate is recolored. If a class $p$ satisfies $\nu^2_{C_s}(p) \neq \nu^1_{C_s}(p)$, then there exists a color region for the class $p$ in $\mathcal{B}_2$ thus there are at most $(k - \omega)$ 
such classes. So we may have to apply the described sequence of recolorings and Lemma \ref{lem:step-2} at most $(k - \omega)$ times. As this sequence recolors each coordinate at most $6$ times and as Lemma \ref{lem:step-2} recolors each coordinate at most $6$ times the result follows.
\end{proof}

Now we apply Lemma~\ref{lem:step-3-same-col-buf} iteratively for $(\mathcal{B}_1,\bnu^1)$ and $(\mathcal{B}_i,\bnu^i)$ where we only recolor vertices in $\mathcal{B}_i \setminus R_{N}^i$. By the separation property, it indeed implies that no vertex of $C$ is recolored. At the end of this procedure, we have recolored vertices of $\mathcal{B}_i$ for $i \ge 2$ at most $12(k-\omega)$ times and the resulting coloring (still denoted by $\bnu^i$ for convenience) is such that $(\mathcal{B}_i, \bnu^i)$ is a valid buffer that satisfies $\nu^1_{C_s}(p) = \nu^i_{C_s}(p)$ and $\nu^1_{C_N}(p) = \nu^i_{C_N}(p)$ for every $p \leq \omega$. In particular, by Observations~\ref{obs:c-buffer-1} and \ref{obs:c-buffer-1-bis}, there exists a color region for class $p$ in $\bnu^i$ if and only if there exists a color region for class $p$ in $\bnu^1$. Moreover the colors that appear and disappear are the same. So in order to be sure that $\nu^1_Q(p)=\nu^i_Q(p)$ for every block $Q$ in $R_s,\ldots,R_N$, we just have to guarantee that the color change for class $p$ are in regions with the same index in $\bnu^i$ and $\bnu^1$. We can guarantee that by iteratively applying Lemma~\ref{lem:move-col-r}. Indeed let $R_j$ be the smallest region of the color buffer where $R_j^1$ and $R_j^i$ do not have the same coloring. If $R_j^1$ is a waiting region, then there exists a waiting region $R_{j'}^2$ in $\{ R_{j+1}^i,\ldots,R_{N-1}^2 \}$ since the number of waiting regions is the same. If $R_j^1$ is a color region where $z$ appears, then by minimality of $i$, there exists $j' > j$ such that $z$ appears in $R^2_{j'}$. By Lemma~\ref{lem:move-col-r}, we can assume that the two colorings $\nu^1$ and $\nu^i$ agree up to region $j$. We will apply at most $(k-\omega)$ times Lemma~\ref{lem:move-col-r}. And then in total, we recolor every vertex of the color buffer at most $(k-\omega)$ times.
\medskip

\noindent\textbf{Step 2:} Agreement on the transposition buffers. \\
Let us first remark that since the proof of Lemma \ref{lem:transp-shift} is symmetric the following holds:

\begin{lemma}\label{lem:shift-transpo-right}
Let $(\mathcal{B}, \bnu)$ be an almost valid buffer and $R_i, R_{i+1}$ be two consecutive regions with $1 < i < s$ such that $R_{i}$ is a transposition region and $R_{i+1}$ is a waiting region. Then there exists a recoloring sequence of $R_i \cup R_{i+1}$ such that in the resulting coloring $\bnu'$, $R_i$ is a waiting region, $R_{i+1}$ is a transposition region and $(\mathcal{B}, \bnu')$ is almost valid. Moreover only coordinates of $R_i, R_{i+1}$ are recolored at most twice.
\end{lemma}

A valid buffer $(\mathcal{B},\bnu)$ is \emph{well-organized} if the first $2 {\omega \choose 2}$ regions of the transposition buffer are waiting regions.

\begin{lemma}\label{lem:well-organised-buf}
Let $(\mathcal{B}, \bnu)$ be a valid buffer. There exists a sequence of recolorings of $R_2, \ldots, R_{s-1}$ such that the resulting coloring $(\mathcal{B}, \bnu')$ is valid and well-organized. Moreover each coordinate is recolored $O(\omega^2)$ times.
\end{lemma}
\begin{proof}
By Lemma~\ref{lem:transp-cancel}, we can assume that (free to recolor at most $O(\omega^2)$ times each vertex of the transposition buffer) there are at most ${\omega \choose 2}$ transposition regions in the transposition buffer of $(\mathcal{B},\bnu)$ and at least $2 {\omega \choose 2}$ waiting regions. By Lemma~\ref{lem:shift-transpo-right}, we can then assume that the first $2 {\omega \choose 2}$ regions of the transposition buffer are waiting regions. Note that performing all these transformations require at most $O(\omega^2)$ recolorings of each coordinate.
\end{proof}
Lemma \ref{lem:well-organised-buf} ensures that we can assume that $(\mathcal{B}_i, \bnu^i)$ is well-organized for every $i \in \llbracket 1, e \rrbracket$.
Let $(\mathcal{B}, \bnu)$ be a valid buffer and $R_j$ with $1 < j < s$ be a region of the transposition buffer. By Observation \ref{obs:t-buffer-1}, $\nu_{A_j}$ and $\nu_{C_j}$ are permutations of $\llbracket 1, \omega \rrbracket$ and there exists a unique transposition $\tau_j$ (which might be the identity) such that $\nu_{C_j} = \tau_j \circ \nu_{A_j}$. Conversely, $\nu_{A_j}$ and $\tau_j$ define a unique coloring of $R_j$ (since the temporary colors used in a valid buffer are fixed). In what follows, for $i \in \llbracket 1, e \rrbracket$ and $j \in \llbracket 2, s-1\rrbracket$, $\tau^i_j$ denotes the transposition corresponding to the region $R^i_j$ of $(\mathcal{B}_i, \bnu^i)$. For the ease of notation, let $\Omega = {{\omega} \choose{2}}$. Note that after step $1$, and since we can assume that all the buffers $(\mathcal{B}_i, \nu^i)$ are well-organized, we have $\prod_{j = s-1}^{2\Omega+2} \tau_j^i = \nu^1_{A_s}$ for all $i \leq e$ (where the symbol $\prod$ denotes the composition). \\
From now on, we will only recolor the transposition buffers of $(\mathcal{B}_i, \bnu^i)$, for $i \ge 2$ to make them agree with $(\mathcal{B}_1, \bnu^1)$. Let us start with the two following lemmas:
\begin{lemma}\label{lem:insert-transpo}
Let $(\mathcal{B}, \bnu)$ be a valid buffer, $t_0,t_1 \in \llbracket 2, s-1 \rrbracket$ with $t_0 < t_1$, and $p \neq q$ be two integers in $\llbracket 1, \omega \rrbracket$. Assume that $R_{t_0}, \ldots, R_{t_1}$ are waiting regions. There exists a sequence of recolorings of $\cup_{j = t_0}^{j = t_1} R_{j}$ such that in the resulting coloring $\bnu'$, $\tau_{t_0} = \tau_{t_1} = \tau_{p,q}$ (and then $R_{t_0}$ and $R_{t_1}$ are transposition regions for the classes $p$ and $q$), for every $t_0 < i < t_1$ $R_i$ is a waiting region, and $(\mathcal{B}, \bnu')$ is valid. Moreover, each coordinate is recolored at most twice.
\end{lemma}
\begin{proof}
Since $(\mathcal{B}, \bnu)$ is valid, $\nu_{A_{t_0}}$ is a permutation of the canonical colors by Observation \ref{obs:t-buffer-1}. As $R_{t_0}, \ldots, R_{t_1}$ are waiting regions, the continuity property ensures that the class $m$ is colored with the same canonical color in $R_{t_0}, \ldots, R_{t_1}$ in $\bnu$. Let $c_1 = \nu_{A_{t_0}}(p)$, $c_2 = \nu_{A_{t_0}}(q)$, and $z \neq z'$ be the (non-canonical) temporary colors used in the transposition buffer of $\mathcal{B}$. We apply the following recolorings:
\begin{enumerate}
    \item Recolor $(B_{t_0}, \ldots, B_{t_1}, p)$ with $z$,
    \item Recolor $(B_{t_0}, q)$ and $(B_{t_1}, q)$ with $z'$,
    \item Recolor $(C_{t_0}, \ldots, A_{t_1}, q)$ with $c_1$,
    \item Recolor $(C_{t_0}, \ldots, A_{t_1}, p)$ with $c_2$.
\end{enumerate}
Recolorings $1$ and $2$ are proper since $R_{t_0}, \ldots, R_{t_1}$ only contain canonical colors in $\bnu$ and $z \neq z'$ are non-canonical. After recoloring $1$, $(C_{t_0}, \ldots, A_{t_1})$ is $c_1$-free and after recolorings $2$ and $3$, $(C_{t_0}, \ldots, A_{t_1})$ is $c_2$-free, thus recolorings $3$ and $4$ are proper. The regions $R_{t_0+1}, \ldots, R_{t_1-1}$ are waiting regions in $\bnu'$ by Observation \ref{obs:maintain-nature-color-waiting-2}. One can then easily check that in $\bnu'$, the regions $R_{t_0}$ and $R_{t_1}$ are transposition regions for the classes $p$ and $q$ and that $(\mathcal{B}, \bnu')$ is valid. 
\end{proof}

\begin{lemma}\label{lem:well-organized-extended}
Let $(\mathcal{B}, \bnu)$ be a valid and well-organized buffer and let 
$T_1, \ldots, T_{\Omega}$  be transpositions of $\llbracket 1, \omega \rrbracket$. There exists a sequence of recolorings of $\cup_{j = 2}^{2\Omega + 1} R_j$ such that in the resulting coloring $\bnu'$ of $\mathcal{B}$, $\tau_{2+j} = \tau_{2\Omega+1-j} = T_{1+j}$ for all $j \in \llbracket 0, \Omega - 1 \rrbracket$ and $(\mathcal{B}, \bnu')$ is valid. Moreover each coordinate is recolored at most $2 \Omega$ times.
\end{lemma}
\begin{figure}
        \centering
        \includegraphics[scale=1.1]{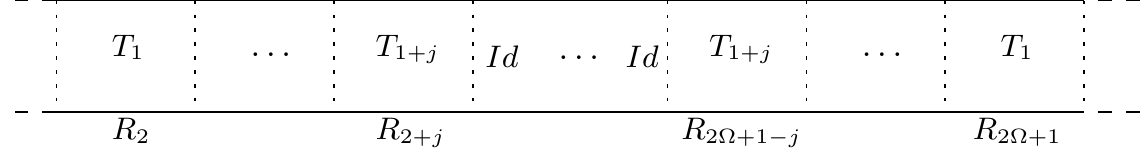}
        \caption{Coloring of regions $R_2, \ldots, R_{2\Omega+1}$ of the buffer $\mathcal{B}$, after iteration $j$ in Lemma \ref{lem:well-organized-extended}.}
        \label{fig:17-2}
\end{figure}
\begin{proof}
For $j \in \llbracket 0, \Omega - 1 \rrbracket$ apply the following:
\begin{itemize}
    \item If $T_{1+j} = Id$, do nothing, else
    \item Let $p \neq q$ be the classes that $T_{1+j}$ permutes. Apply Lemma \ref{lem:insert-transpo} to $\mathcal{B}$ with $t_0 = 2+j$, $t_1 = 2\Omega+1-j$ and integers $p$ and $q$.
\end{itemize}
As $(\mathcal{B}, \bnu)$ is well-organized, the regions $R_2, \ldots, R_{2\Omega + 1}$ are initially waiting regions. Then, after iteration $j$, we obtain a coloring of $\mathcal{B}$ such that for all $t \leq j$, $\tau_{2+t} = \tau_{2\Omega+1-t} = T_{1+j}$ and the regions $R_{2\Omega+j+2}, \ldots, R_{2\Omega-j}$ are waiting regions  (see figure \ref{fig:17-2} for an illustration). Thus we can iterate and the algorithm terminates with the desired coloring $\bnu'$. As Lemma \ref{lem:insert-transpo} maintains a valid buffer, ($\mathcal{B}$, $\bnu'$) is valid. We applied Lemma \ref{lem:insert-transpo} at most $\Omega$ times, thus each coordinate is recolored at most $2\Omega$ times.
\end{proof}

For every $i \in \llbracket 2, \omega \rrbracket$, we apply Lemma \ref{lem:well-organized-extended} to $(\mathcal{B}_i, \bnu^i)$ with $T_{1+j} = \tau^1_{2\Omega+2+j}$ for $j \in \llbracket 0, \Omega - 1 \rrbracket$. Recall that $\nu^1_{A_s} = \prod_{j = s-1}^{2\Omega + 2} \tau_j^1$ and that $(\nu^1_{A_s})^{-1} = \prod_{j = 2\Omega + 2}^{s-1} \tau_j^1$. Then after applying Lemma \ref{lem:well-organized-extended} to $\mathcal{B}_i$ with $i \geq 2$ we obtain a coloring $\bnu^i$ such that $\prod_{j=s-1}^{2} \tau^i_j = \nu^i_{A_s}\circ (\nu^1_{A_s})^{-1} \circ \nu^1_{A_s}$. We will show that we can use this coloring to "cancel" the initial transpositions $\tau_j^i$ (with $j \geq 2\Omega+2$) of $\mathcal{B}_i$. To do so, we need the following lemma:

\begin{lemma}\label{lem:switch-transpo}
Let $R_i, R_{i+1}$ be two consecutive regions of a valid buffer $(\mathcal{B}, \bnu)$ with $1 < i < s$ such that $\tau_{i+1} \neq Id$ and $a$ is a class permuted by $\tau_{i+1}$. There exists a recoloring sequence of $R_i \cup R_{i+1}$ such that in the resulting coloring $\bnu'$, $\tau'_i$ is either the identity or permutes the class $a$ and $\tau'_{i+1}$ does not permute the class $a$. Moreover, $(\mathcal{B}, \bnu')$ is valid and each coordinate is recolored at most $4$ times.
\end{lemma}
\begin{proof}
Let $z,z'$ be the temporary colors used in $(\mathcal{B}, \bnu)$ and let $z'' \notin\{z,z'\}$ be another non-canonical color. We have different cases to consider:
\begin{enumerate}
    \item $\tau_i = Id$. Then $R_i$ is a waiting region. Let $b \neq a$ be the other class permuted by $\tau_{i+1}$. We can apply Lemma \ref{lem:transp-shift} and obtain a valid buffer $(\mathcal{B}, \bnu')$ such that $R_i$ is a transposition region for classes $a$ and $b$ and $R_{i+1}$ is a waiting region. Moreover, the recolorings of Lemma \ref{lem:transp-shift} recolors each coordinate at most twice. 
    
    \item $\tau_i = \tau_{i+1}$. The regions $R_i$ and $R_{i+1}$ are consecutive and permute the same classes, thus they also permute the same colors. By Lemma \ref{lem:transp-cancel} we can recolor $R_{i},R_{i+1}$ into waiting regions and obtain a valid buffer.
    Moreover each coordinate is recolored at most twice.
    
    \item $\tau_i, \tau_{i+1}$ are transpositions which permute exactly one common class. Let  $b \neq c$ distinct from $a$ be the other classes that are permuted by $\tau_i, \tau_{i+1}$. Let $c_1 =\nu_{A_i}(a)$, $c_2 = \nu_{A_i}(b)$ and $c_3 = \nu_{A_i}(c)$. As the two transpositions only involve $\{a, b, c\}$, one of the three classes is permuted in $R_i$ and $R_{i+1}$, and by Lemma \ref{obs:choose-temporary} we can suppose this class is colored with $z$ on $B_i$ and $B_{i+1}$. As $z''$ is not contained in the transposition buffer we can recolor this class on $(B_i, \ldots, B_{i+1})$ with $z''$. The only remaining temporary colors on $R_{i},R_{i+1}$ is $z'$. Thus we can recolor one of the two other class with $z$ on $(B_i, \ldots, B_{i+1})$ and then recolor the third one with $z'$ on the same set of blocks, and these three recolorings are proper. Let $\bnu^{tmp}$ be the coloring we obtain.
    The colors $c_1, c_2, c_3$ are not in $(B_i, \ldots, B_{i+1})$ in $\bnu^{tmp}$. Suppose that $\nu^{tmp}_{C_{i+1}}(a) = c_2$. As $C_{i+1}$ has not been recolored we have necessarily $\nu^{tmp}_{C_{i+1}}(b) = c_3$ and $\nu^{tmp}_{C_{i+1}}(c) = c_1$. We apply the following recolorings:
    \begin{enumerate}
        \item Recolor $(C_i, \ldots, B_{i+1}, a)$ with $c_2$, 
        \item Recolor $(C_i, A_{i+1},b)$ with $c_1$,
        \item Recolor $(B_i, \ldots, A_{i+1}, c)$ with $c_3$.
    \end{enumerate}
    Given that colors $c_1, c_2, c_3$ are not contained in $(B_i, \ldots, B_{i+1})$ in $\bnu^{tmp}$ and given  $\nu^{tmp}_{A_i}$ and $\nu^{tmp}_{C_{i+1}}$ these recolorings are proper by the separation property. Furthermore, in the resulting coloring $\bnu'$, $R_i, R_{i+1}$ are transposition regions and $\tau'_i = \tau_{a,b}$ and $\tau'_{i+1} = \tau_{b,c}$. The case $\nu^{tmp}_{C_{i+1}}(a) = c_3$ is symmetrical.
    \item $\tau_i, \tau_{i+1}$ involves four different classes $a,b,c,d \leq \omega$. We can suppose w.l.o.g that  $\tau_i\tau_{i+1} = \tau_{c,d}\tau_{a,b}$ and by Lemma \ref{obs:choose-temporary} we can assume that $(B_i, c) = (B_{i+1}, a) = z$ and $(B_i, d) = (B_{i+1}, b) = z'$. Let $c_1 = \nu_{A_i}(a)$, $c_2 = \nu_{A_i}(b)$, $c_3 = \nu_{A_i}(c)$, $c_4 = \nu_{A_i}(d)$. We apply the following recolorings:
    \begin{enumerate}
        \item Recolor $(B_i, \ldots, B_{i+1}, a)$ with $z''$,
        \item Recolor $(C_i, \ldots, B_{i+1}, c)$ with $z$,
        \item Recolor $(C_i, \ldots, B_{i+1}, b)$ with $c_1$,
        \item Recolor $(C_i, \ldots, B_{i+1}, d)$ with $z'$,
        \item Recolor $(B_i, \ldots, A_{i+1}, c)$ with $c_3$,
        \item Recolor $(B_i, \ldots, A_{i+1}, d)$ with $c_4$,
        \item Recolor $(B_i, b)$ with $z$,
        \item Recolor $(C_i, \ldots, B_{i+1}, a)$ with $c_2$.
    \end{enumerate}
    Let us justify these recolorings are proper. Recoloring (a) is valid since $z''$ is not contained in the transposition buffer in $\bnu$. Recolorings (b) and (c) are proper since after recoloring (a), the only class colored with $z$ in $R_i, R_{i+1}$ is the class $c$ and the only class colored with $c_1$ on $(B_i, \ldots, C_{i+1})$ is the class $b$. Recoloring (d) is proper since after recoloring (c), the only class colored with $z'$ on $R_i,R_{i+1}$ is the class $d$. At this point the color $c_3$ (resp. $c_4$) is only contained in $(A_i, c)$ and $(C_{i+1}, d)$ (resp. $(A_i, d)$ and  $(C_{i+1}, c)$) in $R_i, R_{i+1}$, thus recolorings $(e)$ and $(f)$ are proper. After recoloring (e), the color $z$ is not contained in $R_i$, thus recoloring (g) is proper. Finally recoloring (f) is proper, since after recoloring (g) the only class colored with $c_2$ on $(B_i, \ldots, C_{i+1})$ is the class $a$. One can then easily check that in $\bnu'$, the region $R_i,R_{i+1}$ are transposition regions, $\tau'_{i} = \tau_{a,b}$ and $\tau'_{i+1} = \tau_{c,d}$.
\end{enumerate}
    Note that in each case, Properties \ref{property-valid-buffer:first-canonical}, \ref{property-valid-buffer:col-buf} and 
    \ref{property-valid-buffer:waiting} of valid buffers are indeed satisfied. Furthermore the continuity property holds as all the given recolorings fit in the framework of Lemma \ref{obs:maintain-continuity}. Property \ref{property-valid-buffer:transp-buf} is satisfied since all the regions of the transposition buffer in $\bnu'$ are either waiting or transposition regions that use the same temporary colors $z$ and $z'$ (up to applying Lemma \ref{obs:choose-temporary} to $R_i$ and $R_{i+1}$), thus $(\mathcal{B}, \bnu')$ is valid. Moreover each coordinate is recolored at most four times.

\end{proof}
We can then give the last lemma before concluding:
\begin{lemma}\label{lem:reduce-well-organised}
Let $(\mathcal{B}, \bnu)$ be a valid buffer. Suppose there exists two integers $t_0, t_1 \in \llbracket 2, s-1 \rrbracket$ with $t_0 < t_1$ such that $\prod_{j = t_1}^{t_0} \tau_j = Id$. Then there exists a sequence of recolorings of $\cup_{j = t_0}^{t_1} R_j$ such that in the resulting colorings $\bnu'$, $R_{t_0}, \ldots, R_{t_1}$ are waiting regions and $(\mathcal{B}, \bnu')$ is valid. Moreover each coordinate is recolored $O(\omega^2)$ times.
\end{lemma}
\begin{proof}
Let $I(\bnu)$ be the number of transpositions equal to the identity in $\tau_{t_0}, \ldots, \tau_{t_1}$ (that is the number of waiting regions in $R_{t_0}, \ldots, R_{t_1}$ for the coloring $\bnu$). If $I(\bnu) = t_1 - t_0 + 1$ we are done. Otherwise we will show that we can always recolor $\cup_{j = t_0}^{t_1} R_j$ and obtain a coloring $\bnu'$ such that the transpositions $\tau'_{t_0}, \ldots, \tau'_{t_1}$ of $R_{t_0}, \ldots, R_{t_1}$ in $\bnu'$ satisfy $\prod_{j = t_1}^{t_0} \tau_j = Id$ and $I(\bnu') > I(\bnu)$. \\
Let $t_0 < r \leq t_1$ be the maximum index such that $\tau_r \neq Id$ and $a$ be a class permuted by $\tau_r$. Let us show that there always exists $i_0 \in \llbracket 0, r - t_0 - 2 \rrbracket$ such that after applying iteratively Lemma \ref{lem:switch-transpo} to regions $R_{r-i-1}, R_{r-i}$ and the class $a$ for $i \in 0, 1, \ldots, i_0$, we obtain a coloring $\bnu''$ such that the transposition $\tau''_{i_0}, \tau''_{i_0-1}$ of $R_{i_0}, R_{i_0-1}$ in $\bnu''$ satisfy $\tau''_{i_0} = \tau''_{i_0-1}$. \\
Suppose not. Then, when applying Lemma \ref{lem:switch-transpo} to regions $R_{r-i-1}, R_{r-i}$ and class $a$ for $i \in 0, 1, \ldots, r-t_0-1$, case $2$ of Lemma \ref{lem:switch-transpo} never occures. Let $\bnu^f$ be the coloring we obtain after the last iteration and $\tau^f$ be the corresponding transpositions. Note that $R_{t_1 + 1}$ is never recolored, thus $T = \prod_{j = t_1}^{t_0} \tau^f_j = Id$. By Lemma \ref{lem:switch-transpo}, $\tau^f_{t_0}$ is the only transposition that permutes the class $a$ amongst $\tau^f_{t_0}, \ldots, \tau^f_{t_1}$ and $\tau^f_{t_0} \neq Id$. Thus $T(a) \neq a$, a contradiction. \\
We obtain a coloring $\bnu''$ such that the transpositions $\tau''_{i_0}, \tau''_{i_0-1}$ of $R_{i_0}, R_{i_0-1}$ in $\bnu''$ satisfy $\tau''_{i_0} = \tau''_{i_0-1}$. By applying Lemma \ref{lem:switch-transpo} once more to $\bnu''$, $R_{i_0}, R_{i_0-1}$ become waiting regions and we obtain a coloring $\bnu'$ such that $I(\bnu') > I(\bnu)$. Furthermore, since $R_{t_1+1}$ is not recolored,  $\prod_{j = t_1}^{t_0} \tau'_j = Id$. \\
We can iterate this process until we obtain a coloring $\bnu'$ (still called $\bnu'$ for convenience) such that $I(\bnu') = t_1-t_0+1$. Furthermore, $(\mathcal{B}, \bnu')$ is valid since Lemma \ref{lem:switch-transpo} maintains a valid buffer. \\

Let us justify the number of times a coordinate is recolored. In order to increase $I(\bnu)$ at each step, we apply Lemma \ref{lem:switch-transpo} to each region of the buffer at most twice, thus each coordinate of the transposition buffer is recolored at most $8$ times. Since $I(\bnu) \leq 3{{\omega}\choose{2}}$, each coordinate of the buffer is recolored $O(\omega^2)$ time.

\end{proof}
\begin{figure}[!h]
  \centering
  \includegraphics[scale=0.9]{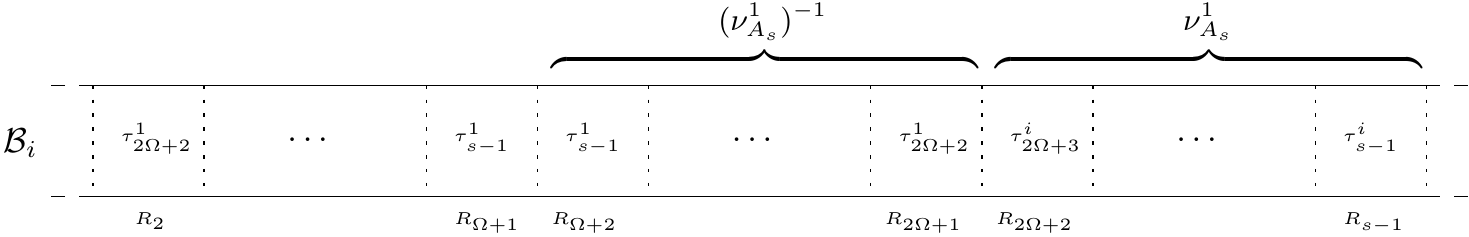}
  \caption{Coloring of the transposition buffer of $\mathcal{B}_i$ with $i \geq 2$, after step $1$ and after applying Lemmas \ref{lem:well-organised-buf} and \ref{lem:well-organized-extended} in step 2. The dotted lines separate the regions.} 
  \label{fig:17-3}
\end{figure}
Let us summarize the step $2$ of the proof. By Lemma \ref{lem:well-organised-buf} we can suppose that all the $(\mathcal{B}_i, \bnu^i)$ are well-organized. Recall that after step $1$, $\nu^i_{A_s} = \nu^1_{A_s}$ for all $i \geq 2$, thus $\prod_{j=s-1}^{2\Omega+2} \tau^i_j = \nu^1_{A_s}$. Then after applying Lemma \ref{lem:well-organized-extended} to $\mathcal{B}_i$ with $i \geq 2$ we obtain a coloring - that we still denote by $\bnu^i$ for convenience - such that for every $j \in \llbracket 2, \Omega+1 \rrbracket$, $\tau^i_j = \tau^1_{2\Omega + j}$ and for all $j \in \llbracket 0, \Omega-1 \rrbracket$, $\tau^i_{\Omega+2 + j} = \tau^1_{s-1-j}$ (see figure \ref{fig:17-3} for an illustration). Thus we have $\prod_{j=s-1}^{\Omega+2} \tau^i_j = Id$ and we can apply Lemma \ref{lem:reduce-well-organised} to $(\mathcal{B}_i, \bnu^i)$ and obtain a valid buffer coloring of $\mathcal{B}_i$ such that for all $j \in \llbracket 2, \Omega+1 \rrbracket$, $\tau^i_j = \tau^1_{2\Omega + j}$, and for every $j \in \llbracket \Omega+2, s-1 \rrbracket$, $R^i_j$ is a waiting region ($\tau^i_j = Id)$. It just remains to switch the transpositions up in the buffer using Lemma \ref{lem:well-organised-buf} to obtain a valid buffer $(\mathcal{B}^i, \bnu^i)$ such that $\bnu^i = \bnu^1$ for every $i \in \llbracket 2, e \rrbracket$. \\
Finally, note that for the proof of Lemma \ref{lem:step-3}, we applied Lemmas \ref{lem:step-3-same-col-buf}, \ref{lem:well-organized-extended} and \ref{lem:reduce-well-organised} once to each $\mathcal{B}_i$ for $i \in \llbracket 2, e \rrbracket$, and we applied Lemma \ref{lem:well-organised-buf} at most twice to each $\mathcal{B}_i$ for $i \in \llbracket 1, e \rrbracket$. Thus every coordinate is recolored $O(\omega^2)$ times, which concludes the proof.


\section{Proof of Lemma \ref{lem:step-4}}\label{ssec:pf-step4}
Let $C$ be a clique of $T$ with children $S_1, S_2, \ldots S_e$ and let $\alpha$ be a $k$-coloring of $G$ treated up to $S_i$ for $i \in 1, \ldots, e$. Let $\nu_C$ be a vector associated with $C$, $\mathcal{B} = R_1, \ldots, R_N$ be the buffer rooted at $C$, and $\mathcal{B}_i = R^i_{1}, \ldots, R^i_{N}$ be the buffer rooted at $S_i$. Suppose that all the $\mathcal{B}_i$s have the same coloring $\bnu$ such that $(\mathcal{B}_i, \bnu)$ is a valid buffer and $D_{\mathcal{B}_i}(\nu_C, \bnu) = 0$.
We will show that there exists a sequence of vertex recolorings of $\cup_{i =1}^e \cup_{j = 2}^{N-1} R^i_j$ such that the resulting coloring of the buffer $\mathcal{B}$ rooted at $C$ is well-colored for $\bnu$. \\ 
Let $Q_1, \ldots, Q_{3\Delta N}$ denote the blocks of $\mathcal{B}$, and $Q^i_1, \ldots, Q^i_{3 \Delta N}$ denote the blocks of $\mathcal{B}_i$ for $i \in \llbracket 1, \omega \rrbracket$. Note that if a vertex starts at height $h$ in $T_{S_i}$ then it starts at height $h+1$ in $T_C$. In particular if a vertex $v$ starts at height $h$ in $T_{S_i}$ such that $j \Delta < h < (j+1)\Delta-1$ for $j \in \llbracket 0, 3N-1 \rrbracket$, then $v \in Q^i_{N-j} \cap Q_{N-j}$. However if a vertex $v$ starts at height $h = j \Delta$, then $v \in Q^i_{N-j} \cap Q_{N-j-1}$, thus we may have to recolor $v$. Indeed, suppose that $Q^i_{N-j}$ is the block $B$ of a color region for the class $p$ and colors $c,z$, and that $v \in (Q^i_{N-j}, p)$. Then $v$ is colored with $\nu_{B}(p) = z$. As we want $\mathcal{B}$ to be well-colored for $\bnu$, $Q_{N-j-1}$ has to be well-colored for $\nu_A$, thus we need to recolor $v$ with $\nu_A(p) = c$. \\ 
For the ease of notation, we denote by $f(Q^i_j)$ the set of vertices of $Q^i_j$ which height in $T_{S_i}$ is maximum, and $f(Q^i_j, p)$ is the set $f(Q^i_j) \cap X_p$ for $p \leq \omega$. 
Let $z,z'$ be the temporary colors used in the transposition buffers of $S_1, \ldots, S_e$. For every $j \in N-1, \ldots, 2$ in the decreasing order, we apply to each region $R^1_j, \ldots R^e_j$ the following recolorings:
\begin{enumerate}
    \item if $R^i_j$ is a waiting region do nothing. By the continuity property and the definition of waiting regions we have $\nu_{C_{i-1}} = \nu_{A_i} = \nu_{B_i} = \nu_{C_i}$, thus $R_i$ is well-colored for $\nu_{A_i}, \nu_{B_i}, \nu_{C_i}$.
    \item\label{proof:shifiting-color-r} if $R^i_j$ is a color region for the class $q$ where $c_1$ disappears, recolor $f(B^i_j, q)$ with $c_1$.\\
    This is a proper recoloring since by Observation \ref{obs:c-buffer-1}, the only class colored with $c_1$ on the color buffer of $\mathcal{B}^i$ is the class $q$. \\
    For any class $m \neq q$, $\nu_{C_{i-1}}(m) = \nu_{A_i}(m) = \nu_{B_i}(m) = \nu_{C_i}(m)$ by definition of color regions and the continuity property. Furthermore, $\nu_{C_{i-1}}(q) = \nu_{A_i}(q) \neq \nu_{B_i}(q) = \nu_{C_i}(q)$. As the vertices $f(B^i_j, q)$ have been recolored with $\nu_{A_i}(q)$ for every $i \leq p$, it follows that $R_i$ is well-colored for $\nu_{A_i}, \nu_{B_i}, \nu_{C_i}$.
    
    \item if $R^i_j$ is a transposition region for the classes $p$ and $q$ with $\nu_{A_j}(p) = c_1$ and $\nu_{C_j}(q) = c_2$ do in the following order:
    \begin{enumerate}
        \item\label{proof:shifting-transpo-c} Recolor $f(C^i_j, p)$ with $z$ and recolor $f(C^i_j, q)$ with $z'$, then
        \item\label{proof:shifting-transpo-a} Recolor $f(B^i_j, p)$ with $c_1$ and recolor $f(B^i_j, q)$ with $c_2$.
    \end{enumerate}
    By the separation property $N[C^i_j] \subseteq B^i_j \cup C^i_j \cup A^i_{j+1}$. As $\bnu$ defines a valid buffer, $R^i_{j+1}$ is either a waiting region or a transposition region, thus $\nu_{A_{i+1}}$ only contains canonical colors. The only vertices colored with $z$ (resp. $z'$) in $N[C^i_j]$ are thus vertices of the class $p$ (resp. $q$), and recoloring (a) is proper. \\
    Furthermore, $N(B^i_j) \subseteq A^i_j \cup B^i_j \cup C^i_j$. By the definition of transposition regions, the only vertices colored with $c_1$ (resp. $c_2$) on $A^i_j \cup B^i_j$ are vertices of the class $p$ (resp. $q$). By the separation property, $N[f(B^i_j)] \cap C^i_j \subseteq f(C^i_j)$, and after recoloring (a) no vertex of $f(C^i_j)$ is colored with $c_1$ nor $c_2$. It follows that recoloring (b) is proper. \\
    For any class $m \notin \{p, q\}$, $\nu_{C_{i-1}}(m) = \nu_{A_i}(m) = \nu_{B_i}(m) = \nu_{C_i}(m)$ by definition of transposition regions and the continuity property. Furthermore, $\nu_{B_i}(p) = z$ (resp. $\nu_{B_i}(q) = z'$) and the vertices $f(C^i_j, p)$ (resp. $f(C^i_j, q)$) have been recolored with $z$ (resp. $z'$). Finally, $\nu_{A_i}(p) = c_1$ (resp. $\nu_{A_i}(q) = c_2$) and the vertices $f(B^i_j, p)$ (resp. $f(B^i_j, q)$) have been recolored with $c_1$ (resp. $c_2$). It follows that the region is well-colored for $\nu_{A_i}, \nu_{B_i}, \nu_{C_i}$.
\end{enumerate}
Let us finally check for the coloring of regions $R_1$ and $R_N$. We have $R_1 \subseteq \cup_{i = 1}^e R^i_1 \cup A^i_2$. As none of these vertices are recolored, and $\nu_{A_1} = \nu_{B_1} = \nu_{C_1} = \nu_{A_2}$ by the continuity property and Property \ref{property-valid-buffer:first-canonical} of valid buffers, $R_1$ is well-colored for $\nu_{A_1}, \nu_{B_1}, \nu_{C_1}$. We have $R_N \subseteq \cup_{i = 1}^e R^i_N \cup C$. None of these vertices are recolored ($R^i_N$ is a waiting region for all $i \leq e$ by Property \ref{property-valid-buffer:waiting} of valid buffers), and by assumption $\nu_C = \nu_{C_N}$. Thus $R_N$ si well-colored for $\nu_{A_N}, \nu_{B_N}, \nu_{C_N}$. It follows that $(\mathcal{B}, \bnu)$ is a valid buffer. As the vertices of height at least $3 \Delta N$ in $T_C$ are not recolored, they remain colored canonically. We obtain a coloring of $G$ treated up to $C$, which concludes the proof.


\begin{thebibliography}{10}

\bibitem{BBRM18}
Hans-Joachim {B{\"o}ckenhauer}, Elisabet {Burjons}, Martin {Raszyk}, and Peter
  {Rossmanith}.
\newblock {Reoptimization of Parameterized Problems}.
\newblock {\em arXiv e-prints}, 2018.

\bibitem{BonamyJ12}
M.~Bonamy, M.~Johnson, I.~Lignos, V.~Patel, and D.~Paulusma.
\newblock {Reconfiguration graphs for vertex colourings of chordal and chordal
  bipartite graphs}.
\newblock {\em Journal of Combinatorial Optimization}, pages 1--12, 2012.
\newblock URL: \url{http://dx.doi.org/10.1007/s10878-012-9490-y}, \href
  {http://dx.doi.org/10.1007/s10878-012-9490-y}
  {\path{doi:10.1007/s10878-012-9490-y}}.

\bibitem{BonamyB17}
Marthe Bonamy and Nicolas Bousquet.
\newblock Token sliding on chordal graphs.
\newblock In {\em Graph-Theoretic Concepts in Computer Science - 43rd
  International Workshop, {WG} 2017}, pages 127--139, 2017.

\bibitem{BonamyB18}
Marthe Bonamy and Nicolas Bousquet.
\newblock Recoloring graphs via tree decompositions.
\newblock {\em Eur. J. Comb.}, 69:200--213, 2018.

\bibitem{BonsmaC07}
P.~Bonsma and L.~Cereceda.
\newblock {Finding Paths Between Graph Colourings: {PSPACE}-Completeness and
  Superpolynomial Distances.}
\newblock In {\em {MFCS}}, volume 4708 of {\em {Lecture Notes in Computer
  Science}}, pages 738--749, 2007.

\bibitem{BoseLPV18}
Prosenjit Bose, Anna Lubiw, Vinayak Pathak, and Sander Verdonschot.
\newblock Flipping edge-labelled triangulations.
\newblock {\em Comput. Geom.}, 68:309--326, 2018.

\bibitem{BousquetHIM18}
Nicolas Bousquet, Tatsuhiko Hatanaka, Takehiro Ito, and Moritz
  M{\"{u}}hlenthaler.
\newblock Shortest reconfiguration of matchings.
\newblock {\em CoRR}, abs/1812.05419, 2018.

\bibitem{BousquetH19}
Nicolas {Bousquet} and Marc {Heinrich}.
\newblock {A polynomial version of Cereceda's conjecture}.
\newblock {\em arXiv e-prints}, 2019.
\newblock \href {http://arxiv.org/abs/1903.05619} {\path{arXiv:1903.05619}}.

\bibitem{BousquetM18}
Nicolas Bousquet and Arnaud Mary.
\newblock Reconfiguration of graphs with connectivity constraints.
\newblock In {\em Approximation and Online Algorithms - 16th International
  Workshop, {WAOA} 2018}, pages 295--309, 2018.
\newblock \href {http://dx.doi.org/10.1007/978-3-030-04693-4\_18}
  {\path{doi:10.1007/978-3-030-04693-4\_18}}.

\bibitem{BousquetP16}
Nicolas Bousquet and Guillem Perarnau.
\newblock Fast recoloring of sparse graphs.
\newblock {\em Eur. J. Comb.}, 52:1--11, 2016.

\bibitem{Cereceda}
L.~Cereceda.
\newblock {\em {Mixing Graph Colourings}}.
\newblock PhD thesis, London School of Economics and Political Science, 2007.

\bibitem{Cereceda09}
L.~Cereceda, J.~van~den Heuvel, and M.~Johnson.
\newblock {Mixing 3-colourings in bipartite graphs}.
\newblock {\em Eur. J. Comb.}, 30(7):1593--1606, 2009.
\newblock URL: \url{http://dx.doi.org/10.1016/j.ejc.2009.03.011}, \href
  {http://dx.doi.org/10.1016/j.ejc.2009.03.011}
  {\path{doi:10.1016/j.ejc.2009.03.011}}.

\bibitem{CerecedaHJ11}
L.~Cereceda, J.~van~den Heuvel, and M.~Johnson.
\newblock {Finding paths between 3-colorings}.
\newblock {\em Journal of Graph Theory}, 67(1):69--82, 2011.
\newblock URL: \url{http://dx.doi.org/10.1002/jgt.20514}, \href
  {http://dx.doi.org/10.1002/jgt.20514} {\path{doi:10.1002/jgt.20514}}.

\bibitem{ChenDMPP19}
Sitan Chen, Michelle Delcourt, Ankur Moitra, Guillem Perarnau, and Luke Postle.
\newblock Improved bounds for randomly sampling colorings via linear
  programming.
\newblock In {\em Proceedings of the Thirtieth Annual {ACM-SIAM} Symposium on
  Discrete Algorithms, {SODA} 2019}, pages 2216--2234, 2019.
\newblock \href {http://dx.doi.org/10.1137/1.9781611975482.134}
  {\path{doi:10.1137/1.9781611975482.134}}.

\bibitem{Diestel}
R.~Diestel.
\newblock {\em {Graph Theory}}, volume 173 of {\em {Graduate Texts in
  Mathematics}}.
\newblock Springer-Verlag, Heidelberg, third edition, 2005.

\bibitem{dyer2006randomly}
M.~Dyer, A.~D. Flaxman, A.~M Frieze, and E.~Vigoda.
\newblock {Randomly coloring sparse random graphs with fewer colors than the
  maximum degree}.
\newblock {\em Random Structures \& Algorithms}, 29(4):450--465, 2006.

\bibitem{article}
Carl Feghali.
\newblock Paths between colourings of graphs with bounded tree-width.
\newblock {\em Information Processing Letters}, 144, 12 2018.
\newblock \href {http://dx.doi.org/10.1016/j.ipl.2018.12.006}
  {\path{doi:10.1016/j.ipl.2018.12.006}}.

\bibitem{FeghaliJP16}
Carl Feghali, Matthew Johnson, and Dani{\"{e}}l Paulusma.
\newblock A reconfigurations analogue of brooks' theorem and its consequences.
\newblock {\em Journal of Graph Theory}, 83(4):340--358, 2016.

\bibitem{10.1007/3-540-60618-1_88}
Philippe Galinier, Michel Habib, and Christophe Paul.
\newblock Chordal graphs and their clique graphs.
\newblock In Manfred Nagl, editor, {\em Graph-Theoretic Concepts in Computer
  Science}, pages 358--371, Berlin, Heidelberg, 1995. Springer Berlin
  Heidelberg.

\bibitem{ItoKK0O17}
Takehiro Ito, Naonori Kakimura, Naoyuki Kamiyama, Yusuke Kobayashi, and Yoshio
  Okamoto.
\newblock Reconfiguration of maximum-weight b-matchings in a graph.
\newblock In {\em Computing and Combinatorics - 23rd International Conference,
  {COCOON} 2017, Hong Kong, China, August 3-5, 2017, Proceedings}, pages
  287--296, 2017.

\bibitem{LokshtanovM18}
Daniel Lokshtanov and Amer~E. Mouawad.
\newblock The complexity of independent set reconfiguration on bipartite
  graphs.
\newblock In {\em Proceedings of the Symposium on Discrete Algorithms, {SODA}
  2018}, pages 185--195, 2018.

\bibitem{mohar4}
Bojan Mohar and Jes{\'u}s Salas.
\newblock On the non-ergodicity of the {S}wendsen--{W}ang--{K}oteck{\'y}
  algorithm on the {K}agom{\'e} lattice.
\newblock {\em Journal of Statistical Mechanics: Theory and Experiment},
  2010(05):P05016, 2010.

\bibitem{Nishimura17}
N.~Nishimura.
\newblock Introduction to reconfiguration.
\newblock {\em preprint}, 2017.

\bibitem{SuzukiMN14}
A.~Suzuki, A.~Mouawad, and N.~Nishimura.
\newblock {Reconfiguration of Dominating Sets}.
\newblock {\em CoRR}, 1401.5714, 2014.
\newblock URL: \url{http://arxiv.org/abs/1401.5714}.

\bibitem{Heuvel13}
J.~van~den Heuvel.
\newblock {\em {The Complexity of change}}, page 409.
\newblock {Part of London Mathematical Society Lecture Note Series}. {S. R.
  Blackburn, S. Gerke, and M. Wildon} edition, 2013.

\end{thebibliography}


\end{document}